\documentclass[a4paper, 11pt]{article}
\usepackage[utf8]{inputenc}
\usepackage{amsthm, booktabs}
\usepackage{amsmath}
\usepackage{amssymb}
\usepackage{graphicx}
\usepackage{rotating}
\usepackage[round]{natbib}
\usepackage{bbm}
\usepackage{textcomp}
\usepackage{mathtools}
\usepackage{url}
\usepackage[margin=2.5cm]{geometry}
\usepackage{mathrsfs} 
\usepackage{xcolor}
\usepackage{csquotes}
\usepackage{float}
\usepackage{pdfpages}
\usepackage{comment}

\usepackage{algorithm}
\usepackage{algpseudocode}

\theoremstyle{plain}  
\newtheorem{thm}{Theorem}[section] 
\newtheorem{lem}[thm]{Lemma} 
\newtheorem{prop}[thm]{Proposition}

\theoremstyle{definition}

\newtheorem{simulation}{Simulation}

\theoremstyle{remark} 
\newtheorem{rem}{Remark}

\theoremstyle{plain} 
\newcommand{\thistheoremname}{}
\newtheorem{genericthm}[thm]{\thistheoremname}

\newtheorem*{genericthm*}{\thistheoremname}
\newenvironment{namedthm*}[1]
{\renewcommand{\thistheoremname}{#1}%
	\begin{genericthm*}}
	{\end{genericthm*}}

\newcommand{\prob}{\mathbb{P}} 
\newcommand{\E}{\mathbb{E}} 
\newcommand{\F}{\mathcal{F}} 
\newcommand{\X}{\mathcal{X}} 
\newcommand{\PR}{\mathcal{P}(\mathbb{R})} 

\DeclareMathOperator*{\argmax}{arg\,max} 
\newcommand{\diff}{\,\mathrm{d}} 

\newcommand{\R}{\mathbb{R}} 
\newcommand{\Q}{\mathbb{Q}} 
\newcommand{\N}{\mathbb{N}} 

\newcommand{\probM}{\mathfrak{B}(\Omega)} 
\renewcommand{\P}{\mathcal{P}} 
\newcommand{\PX}{\mathcal{P}(\mathcal{X})} 

\newcommand{\one}{\mathbbm{1}} 

\newcommand{\M}{\mathcal{M}} 
\renewcommand{\H}{\mathcal{H}} 
\newcommand{\crps}{\textrm{CRPS}}  

\newcommand{\bA}{\boldsymbol{A}}
\newcommand{\bB}{\boldsymbol{B}}
\newcommand{\bC}{\boldsymbol{C}}

\newcommand{\bX}{\boldsymbol{X}}
\newcommand{\bDelta}{\boldsymbol{\Delta}}
\newcommand{\diagM}{\R_{0}^{m \times m}}

\newcommand{\eqand}{\quad \textrm{and}\quad } 

\newcommand{\myS}{{\rm S}} 
\newcommand{\QS}{\textrm{QS}} 
\newcommand{\BS}{\textrm{BS}} 

\newcommand{\eqforall}{\quad \textrm{for all }}

\newcommand{\ssupM}{\mathcal{M}^{\textrm{s}, \star}} 
\newcommand{\wsupM}{\mathcal{M}^{\textrm{w}, \star}} 
\newcommand{\supM}{\mathcal{M}^{\star}} 
\newcommand{\uwsupM}{\mathcal{M}^{\textrm{uw}, \star}}

\begin{document}
	\title{Sequential model confidence sets}
	
	\author{Sebastian Arnold\thanks{Both authors contributed equally to this work.} \thanks{Centrum Wiskunde \& Informatica (CWI), \texttt{sebastian.arnold@cwi.nl}} , Georgios Gavrilopoulos\footnotemark[1] \thanks{Eidgenössische Technische Hochschule Zürich (ETHZ), \texttt{ georgios.gavrilopoulos@stat.math.ethz.ch; ziegel@stat.math.ethz.ch}}, Benedikt Schulz\thanks{Karlsruhe Institute of Technology (KIT), \texttt{benedikt.schulz2@kit.edu}}, Johanna Ziegel\footnotemark[3]}

	\maketitle

	\begin{abstract}
		In most prediction and estimation situations, scientists consider various statistical models for the same problem, and naturally want to select amongst the best.
		Hansen et al.\ (2011) provide a powerful solution to this problem by the so-called model confidence set, a subset of the original set of available models that contains the best models with a given level of confidence. Importantly, model confidence sets respect the underlying selection uncertainty by being flexible in size. However, they presuppose a fixed sample size which stands in contrast to the fact that model selection and forecast evaluation are inherently sequential tasks where we successively collect new data and 
		where the decision to continue or conclude a study may depend on the previous outcomes.
		In this article, we extend model confidence sets sequentially over time by relying on sequential testing methods through e-processes and confidence sequences. 
		Sequential model confidence sets allow to continuously monitor the models' performances and come with time-uniform, nonasymptotic coverage guarantees.
	\end{abstract}
	
	
	\section{Introduction}
	
	In science as well as in our daily life, we frequently encounter situations in which multiple statistical models or forecasts are available for the same problem and where we have to decide which model(s) or forecast(s) we want to trust. We might think for example of multiple regression models for inflation with respect to different sets of covariates or of multiple weather services which issue precipitation predictions for the next day. 
	In such situations, we naturally want to select the best model(s) or forecast(s), where the term \enquote{best} is defined in terms of a user-specified criterion, typically given by some real-valued loss function where lower scores correspond to better performance.
	
	The \emph{model confidence set (MCS)} proposed by \cite{MCS} provides a promising solution to this problem. It departs from standard practice where just a single model is selected according to some appropriate loss and nothing is said about the uncertainty associated with this selection. The MCS takes this uncertainty into account by reducing the original set of available models to a smaller set of flexible size that contains the best models with a given level of confidence.
	Model confidence sets are of great importance in applications without an obvious benchmark, may be easily constructed by the MCS algorithm, and are widely applied in the econometrics community; see, e.g., \cite{Nan} and \cite{peng2018best} amongst others.
	
	However, even though the MCS provides an appealing solution to a highly important problem, it comes with limitations. In particular, it is assumed that the sequences of observed loss differences are stationary, an assumption which may be questionable in practice since we may expect that the models evolve over time, and correct for errors and systematic biases by using past information. Apart from relying on strong stationarity assumptions, the MCS requires a sample of some fixed size chosen independently of the data.
	In other words, we assume that we observe losses of the models over some prespecified evaluation period in order to compute the MCS only once at the end of the period. This procedure highly contrasts our natural urge to assess the models on a regular basis by successively including past observations to the available dataset. Consider for example different weather institutions that predict the accumulated precipitation for the following day on a daily-basis.
	Then, firstly, it is not given that the initially better models remain superior until the end of our study.  
	Secondly, we want to assess the forecasters and monitor their performance sequentially, say, at the end of each day, week or month, since the forecasters issue their forecasts sequentially as well and may behave nonstationarily. If we gather evidence that one institution is worse than the others but the evidence is not yet statistically significant, then we want to continue to collect evidence for this hypothesis without losing the information of the earlier observations. Crucially, we want to decide at the end of each day, week or month whether to continue or conclude the study, depending on the previous outcomes.
	Finally, the derivation of the MCS is based on bootstrap approximations and asymptotic, CLT-related properties.
	In addition to high computational costs, coverage guarantees are only asymptotic.
	
	Selecting models and comparing forecasts are inherently sequential tasks, and we should seek methods that allow for a sequential choice of the best models while respecting the underlying selection uncertainty. In this article, we address this need and contribute with sequential model confidence sets. 
	Sequential model confidence sets generalize the core idea of \cite{MCS} sequentially over time and enjoy time-uniform coverage guarantees of the theoretically superior objects. 
	Moreover, they have non-asymptotic validity and do not rely on computationally expensive approximations.
	Our methods rely on e-processes and time-uniform confidence sequences, which 
	lay the foundation for safe anytime-valid inference, a field which met a surge of new contributions over the last years by various authors; see, e.g., \cite{Shafer2021}, \cite{Vovk_Wang_2021}, \cite{ramdas2023gametheoretic} and \cite{GrunwaldHeideETAL2019} amongst others. 
	
	Sequential testing methods relying on minimal distributional assumptions have been used before for sequential forecast evaluation: \cite{HenziZiegel2021} provide e-processes to sequentially test the strong hypothesis that one forecaster is better than another at all time points.
	In contrast, \cite{Choe_Ramdas_2021} develop e-processes for the weaker hypothesis that one forecaster is uniformly better than the other on average.
	Both studies assess the performance of some forecaster at time $t\in \N$ with respect to the (average) expected loss difference given all information up to the previous time step.
	This is relevant to tests of equal conditional predictive ability proposed by \cite{Giacomnini_White_2006}.
	In contrast to the unconditional approach, most prominently in the seminal work of \cite{Diebold_Mariano_1995}, the conditional approach uses the available information to infer which forecaster is more accurate on a specific date.
	In this article, we follow the conditional approach and study the strong as well as the weak hypothesis.
	In contrast to all aforementioned contributions on forecast comparison, our methods allow to simultaneously draw inference on the performance of $m\geq 2$ different forecasters rather than only being able to compare two different ones. 
	
	In order to construct sequential model confidence sets, we make use of the e-processes provided by \cite{Howard_Chernoff_bounds} and \cite{Choe_Ramdas_2021} under the assumption of sub-Gaussian or sub-exponential loss differences, which subsumes in particular the case where the loss differences are conditionally bounded, an assumption, which does not only hold for bounded losses but also for prominent choices of scoring rules such as the continuous ranked probability score (CRPS; \citealp{Matheson1976}) or the quantile score \citep{Gneiting2011a}.

	The remainder of the paper is structured as follows. Section \ref{Sec:SMCS} presents the problem and introduces sequential model confidence sets. In Section \ref{Sec:Construction_of_SMCS}, we review important definitions from sequential testing theory and provide methods to construct sequential model confidence sets with respect to three particular notions of superior objects. A simulation study is conducted in Section \ref{Sec:Simulation} before we apply the proposed methods in case studies on Covid-19 related deaths and wind gust predictions in Section \ref{Sec:Case_study}. The main part of the paper closes with a discussion in Section \ref{Sec:Discussion}. Some background on forecast evaluation, technical comments, and proofs are available in the supplementary material.

	\section{Sequential model confidence sets}\label{Sec:SMCS}
	We consider a set $\M_0=\{1, \dots, m\}$ indexing $m\geq 2$ different statistical models or forecasters, simply referred to as \emph{models}, and we let $L$ be a loss function which measures the quality of the models and is negatively oriented, that is, lower values correspond to better performance. We collect data in discrete time and denote the loss of model $i\in \M_0$ at time $t\in \N$ by $L_{i,t}$. We assume that all random quantities are defined on some underlying measurable space $(\Omega, \F)$ equipped with some filtration $(\F_t)_{t \in \N}$. 
	We denote the family of all probability measures on $(\Omega, \F)$ by $\probM$ and write $\Q$ for a generic element in $\probM$. 
	
	Our leading example of the above setting is the following. Let $1, \dots, m$ index different forecasters which all sequentially issue predictive distributions or point forecasts $(f_{1,t})_{t \in \N}, \dots ,(f_{m,t})_{t \in \N}$ for some unknown quantity $(Y_t)_{t \in \N}$, adapted to \((\mathcal{F}_t)_{t\in \N}\).
	For each $i=1, \dots, m$ and $t\in \N$, the forecast $f_{i,t}$, which refers to the outcome \(Y_t\), is \emph{predictable}, meaning that it is only based on information until time \(t-1\). To assess the forecasters, one should employ \emph{proper scoring rules} 
	or \emph{consistent scoring functions}, see Supplement \ref{appendix:forecast_evaluation} for further details. Then, the time series $(L_{i,t})_{t \in \N}$ emerges by applying a proper scoring rule (consistent scoring function) $L$ to the forecast-observation pairs $(f_{i,t},Y_t)_{t \in \N}$, that is $L_{i,t}=L(f_{i,t},Y_t)$, $i=1, \dots, m$, $t\in \N$. 
	
	Next, we clarify the notion of ``best models''. For $i, j\in \M_0$ and $t\in \N$, consider $d_{ij,t}=L_{i,t}-L_{j,t}$ and $\mu_{ij, t}=\E(d_{ij,t} \mid \F_{t-1} )$ as well as their averages $\hat{\Delta}_{ij,t}= \sum_{s=1}^t d_{ij, s}/t$ and $\Delta_{ij,t}=\sum_{s=1}^t \mu_{ij, s}/t$, where we omit the dependence of $\mu_{ij,t}$ and $\Delta_{ij,t}$ on the choice of the probability measure $\Q$ on $(\Omega, \F)$ for the sake of brevity.
	
	We follow \cite{Giacomnini_White_2006}, \cite{Lai_et_al_2011}, \cite{HenziZiegel2021} and \cite{Choe_Ramdas_2021} by defining the superior models in terms of the (average) conditional expected loss differences. 
	In particular, we define the \emph{strongly superior objects} \begin{equation}
		\label{eq:def_strongly_superior_objects}
		\ssupM= \left\{i \in \M_0 \mid \mu_{ij,t}\leq 0  \textrm{ a.s.~for all }  j \in \M_0,\, t \in \N\right\}, 
	\end{equation}
	the \emph{uniformly weakly superior objects}
	\begin{equation}
		\label{eq:def_uniformly_weakly_superior_objects}
		\uwsupM= \left\{i \in \M_0 \mid \Delta_{ij,t}\leq 0 \textrm{ a.s.~for all } j \in \M_0,\, t \in \N\right\}, 
	\end{equation}
	and the \emph{weakly superior objects}\begin{equation}\label{eq:def_weakly_superior_objects}
		\wsupM_t= \left\{i \in \M_0 \mid \Delta_{ij,t}\leq 0 \textrm{ a.s.~for all } j \in \M_0\right\}, \quad t\in \N.
	\end{equation}
	For each $t\in \N$, we have $\ssupM \subseteq \uwsupM  \subseteq \wsupM_t$. Importantly, $\wsupM_t \neq \emptyset $ for all $t\in \N$, whereas $\ssupM$ and $\uwsupM$ may be empty.
	
	Let $(\supM_t)_{t\in \N}\subseteq \M_0$ be the targeted sequence of superior objects. In the case of \eqref{eq:def_strongly_superior_objects} or \eqref{eq:def_uniformly_weakly_superior_objects}, we just consider a constant sequence. For a given confidence level $\alpha \in (0,1)$, we call $(\widehat{\M}_t)_{t \in \N} \subseteq \M_0$ a \emph{sequence of model confidence sets} or \emph{sequential model confidence sets}, for short SMCS(s), for $(\supM_t)_{t\in \N}$ at level $\alpha$ if, for any $\Q\in \probM$, $\Q(\exists t\geq 1:  \supM_t  \nsubseteq \widehat{\M}_t ) \leq \alpha$, or equivalently 
	\begin{equation}\label{eq:def_coverage}
		\Q(\forall t\geq 1: \supM_t  \subseteq \widehat{\M}_t  ) \geq 1-\alpha.
	\end{equation}
	We refer to this property as a \emph{time-uniform coverage guarantee} for the sequential model confidence sets and highlight the fact that the time quantifiers are inside the probability, which is a much stronger requirement than assuming $\Q(\supM_t  \subseteq \widehat{\M}_t) \geq 1-\alpha$ for all $t\in \N$ (or assuming that there exists a $t\in \N$ with $\Q(  \supM_t  \nsubseteq \widehat{\M}_t ) \leq \alpha$); see, e.g., \citet[Section 1.1]{Howard_confidence_sequences}, for a historical perspective on confidence sequences.
	
	In the paper by \cite{MCS}, the authors make the assumption that the (unconditional) expectation \(\E(d_{ij,t})= \mu_{ij}\) does not depend on \(t\), and they define the set of superior objects as $\M^\star = \left\{i\in \M_0\mid \mu_{ij}\leq 0 \text{ for all } j\in \M_0\right\}.$ This is in line with \cite{Diebold_Mariano_1995}, who also define model (forecast) superiority with respect to the unconditional expected loss differences. For a fixed sample size, they target an asymptotic guarantee of the form \(\Q(\M^\star\subseteq \widehat{\M})\geq 1-\alpha\) for their model confidence set (MCS) \(\widehat{\M}\). Our method provides a stronger guarantee but it is also the first method with finite sample validity for a fixed sample size.
	
	Model confidence sets are also related to the literature on inference on the argmin of a vector of means, see, e.g. \citet{FutschikPflug1995,ZhangLeeETAL2024}. In this strand of literature, the desired coverage guarantee is typically weaker in that $\Q(i \in \widehat{\M}) \ge 1-\alpha$ should hold for each superior model $i$ but not uniformly over all models. We discuss an adaptation of our methods to this so-called \emph{marginal coverage} guarantee in Supplement \ref{appendix:marginal_coverage}.

	\subsection{Discussion of the different sets of superior models}\label{Subsec:Discussion_of_superior_models}
	
	Studying different notions of superior models is important, since, depending on the particular situation, one might interpret the term \enquote{best} differently.
	In this subsection, we discuss the hypotheses which correspond to the sequences of superior models given at \eqref{eq:def_strongly_superior_objects}, \eqref{eq:def_uniformly_weakly_superior_objects}, and \eqref{eq:def_weakly_superior_objects}, respectively, and give concrete examples for each of them.

	If it is reasonable to assume that some models outperform all other models at all time points (in terms of the conditional expected score differences), then we should try to find estimators for $\ssupM$. We refer to this assumption as the \emph{strong hypothesis}. 
	The strong hypothesis particularly applies if we assume independent and identically distributed scores. Then, one typically studies the unconditional expected score differences $\E(d_{ij,t})$ independent of $t\in \N$, as it is done by \cite{Diebold_Mariano_1995} or \cite{MCS}, who assume stationary performances of the models. 
	Another important instance where the strong hypothesis is reasonable is if the models have nested information sets and use them in a nearly optimal way, see \cite{Holzmann_Eulert}.
	In Section \ref{Subsec:Covid}, we construct SMCSs for different Covid-19 related deaths where the assumption of nested information sets seems plausible. 
	
	On the other hand, if we assume that some models have a lower conditional expected score on average rather than at all single time points, then we should target $\uwsupM$ under the \emph{uniformly weak hypothesis}.
	\citet[Section 4.4.]{Choe_Ramdas_2021} 
	argue convincingly that there are many situations where testing for the strong hypothesis may be misleading.
	Inspired by their arguments, Simulation \ref{example:simulation2} below considers a forecaster which is ideal on most days and only slightly worse than the other forecasters on Sundays.
	Any powerful method to monitor the strongly superior models over time would exclude this forecaster from the SMCS. Nevertheless, the forecaster is superior under the uniformly weak hypothesis. 
	
	Finally, the \emph{weak hypothesis} applies if we expect that the models evolve and their average relative performance might change over time.
	In Section \ref{Subsec:Windgust}, we study post-processing methods for wind gust predictions. 
	The results show that in this application it is indeed beneficial to monitor the weakly superior objects, as some models which are excluded from the SMCS at earlier time points are included again at later stages, either due to systematic changes in the underlying meteorological model, or due to the different adaptive behaviour of the methods. If we had conducted the study under any of the two stronger hypotheses, we could not have observed that some methods become again competitive towards the end of the test period. 
	
	We introduce the simulation settings of Section \ref{Sec:Simulation} to give concrete examples for each of the discussed hypotheses. 
	
	\begin{simulation}\label{example:simulation1}  We sample \((Y_t)_{t=1}^{n}\), where $Y_{t}\sim \mathcal{N}(Y_{t-1},1),$ for $t =1, \dots, n$ and $Y_0=0$. We consider $m=49$ different forecasters $\{1, \dots, m\}$ which all sequentially issue predictive distributions for \((Y_t)_{t=1}^{n}\) given by $f_{i,t}=\mathcal{N}(Y_{t-1}+\varepsilon_i,1+\delta_i)$, for $(\varepsilon_i,\delta_i) \in \{-0.6,-0.4,-0.2,0,0.2,0.4, 0.6\}^2$. For these forecasters with different biases and dispersion errors, there is exactly one ideal forecaster ${i_0}$ with $\varepsilon_{i_0}=\delta_{i_0}=0$. We assess the different forecasters with respect to the \emph{continuous ranked probability score} (CRPS; \cite{Matheson1976}).
		The CRPS is a proper scoring rule implying that the ideal forecaster has the lowest expected CRPS.
		According to the definition given at \eqref{eq:def_strongly_superior_objects}, this yields that $\ssupM = \{i_0\}$.
		A brief introduction to (proper) scoring rules, including information about the CRPS, can be found in Supplement \ref{appendix:forecast_evaluation}.
	\end{simulation}
	
	\begin{simulation}\label{example:simulation2}
		We consider the same data-generating mechanism and the same forecasters as in Example \ref{example:simulation1} except that, for $t\in 7\N $, the forecaster ${i_0}$ now issues the forecasts $f_{i_0,t}=\mathcal{N}\left(Y_{t-1}+0.3,1.3\right)$. That is, forecaster $f_{i_0}$ is still ideal on, say, weekdays and Saturdays, however, on Sundays, some forecasters have a smaller bias and dispersion error, and hence $\ssupM=\emptyset$. However, ${i_0}$ is uniformly weakly superior, that is $\uwsupM=\{{i_0}\}$.
	\end{simulation}
	
	\begin{simulation}\label{example:simulation3}
		We sample i.i.d.\ standard normally distributed observations $(Y_t)_{t=1}^n$ and compare $m=3$ different forecasters $i=1,2,3$ which issue median predictions $m_{i,t}= Y_t+\varepsilon_{i,t}$, for $\varepsilon_{1,t}= \beta,  \varepsilon_{2,t}= \gamma^t, \varepsilon_{3,t}= \delta t,$ for $\beta, \delta >0$ and $0<\gamma<1$. We call forecaster $1$ \emph{constantly biased}, forecaster $2$ \emph{improving} and forecaster $3$ \emph{worsening}. We assess them by $L(m,y)=0.5(\Phi(m)-\Phi(y))$, for $m,y\in \R$ and $\Phi$ the cdf of the standard normal distribution, which is a consistent scoring function for the median, see Supplement \ref{appendix:forecast_evaluation}. If we choose the parameters as given in Section \ref{Sec:Simulation}, we have $\wsupM_t = \{3\}$ for $t\leq 153$, $\wsupM_t = \{1\}$ for $153 < t< 550$, and $\wsupM_t = \{2\}$ for $t\geq 550$. That is, whereas the worsening forecaster is initially superior, the constantly biased forecaster catches up after some time until the improving forecaster becomes the best in the end.    
	\end{simulation}

	\section{Construction of sequential model confidence sets}\label{Sec:Construction_of_SMCS}
	
	In this section, we provide methods to construct SMCSs for the superior models given at \eqref{eq:def_strongly_superior_objects}, \eqref{eq:def_uniformly_weakly_superior_objects} and \eqref{eq:def_weakly_superior_objects}. Our constructions build on e-processes, confidence sequences and sequential multiple testing methods. 
	
	\subsection{Sequential testing methods}
	
	Let  $(\Omega, \F)$ be a measurable space and $\probM$ be the family of all probability measures on $(\Omega, \F)$. Suppose that we observe random outcomes of a process at time points \(t=1,2,\ldots\). A \emph{statistical (null) hypothesis} $\H \subseteq \probM$ is a set of probability measures that are potential candidates for the true probability measure $\prob$ governing the data generating process.
	For example, $\H$ could consist of all probability measures under which the data points are i.i.d.~normally distributed with some given mean and variance.
	
	Let $(\F_t)_{t \in \N}$ be a filtration, that is, an increasing sequence of \(\sigma\)-algebras on \((\Omega,\mathcal{F})\), which are all contained in \(\mathcal{F}\).
	We interpret $\F_t$ as the available information at time $t\in \N$.
	
	A \emph{(super-)martingale} with respect to a probability measure \(\Q \in \probM\) is a sequence of \(\Q\)-integrable random variables \((X_t)_{t\in \N}\), which is adapted to \((\F_t)_{t\in \N}\) and satisfies the condition \(\E_\Q[X_t|\F_{t-1}]=X_{t-1}\) (\(\E_\Q[X_t|\F_{t-1}]\leq X_{t-1}\)), for all \(t\in \N\).
	We call a process $(M_t)_{t \in \N}$ a \emph{test (super-)martingale} for some hypothesis $\H\subseteq \probM$ if it is a nonnegative (super)martingale with respect to any \(\Q\in \H\), and has a starting value \(M_0\leq 1\).
	
	Test supermartingales allow to phrase statistical testing as a betting game. Indeed, the value of the test supermartingale can be interpreted as the accumulated wealth of a gambler after having bet a number of times against the null hypothesis \citep{Shafer2021}.
	For example, in a coin-toss game with probability of heads equal to \(q\in (0,1)\), a player that bets 1\$ on heads in each round and receives 2\$ in case of success or 0\$ otherwise, cannot expect to gain money over time if the null hypothesis \(\H_0:q\leq 1/2\) is true.
	Thus, large values of the test supermartingale give evidence against the null hypothesis.
	
	Ville's inequality \citep{Ville} states that any test supermartingale $(M_t)_{t \in \N}$ for $\H\subseteq \probM$ satisfies
	\begin{equation}\label{eq:Ville}
		\sup_{\Q \in \H} \Q\left(\exists t \in \N: M_t \geq \frac{1}{\alpha}\right) \leq \alpha, \eqforall \alpha \in (0,1).
	\end{equation}
	This lays the foundation for safe anytime-valid inference.
	A proof can be found in \cite[Section~6.1]{Howard_Chernoff_bounds}.
	By \eqref{eq:Ville}, test supermartingales yield valid sequential tests $(\psi_t)_{t \in \N}$ by defining $\psi_t = \one\{M_t \ge 1/\alpha\}$.
	Furthermore, Ville's inequality enables the construction of time-uniform confidence sequences (\citealp{Howard_confidence_sequences}; \citealp{ramdas2023gametheoretic}). 
	For $\alpha \in (0,1)$, a $(1-\alpha)$-confidence sequence for a parameter sequence $(\theta_t)_{t\in \N}$ in some space $\Theta$ is a sequence of sets $(C_t)_{t\in \N}$ such that $\Q\left(\forall t \in \N: \theta_t \in C_t\right)\geq 1-\alpha$, for all $\Q \in \probM.$
	The definition of confidence sequences resembles that of confidence sets, with the important difference that confidence sequences provide time-uniform coverage guarantees.
	We will use confidence sequences to construct SMCSs for the sequence of weakly superior objects.
	
	It has been shown that for some hypotheses \(\H\), it is not possible to construct non-trivial test supermartingales \citep{Ramdas2022eprocess, henzi&law}. However, it may still be possible to find non-trivial e-processes.
	An \emph{e-process} for some hypothesis $\H\subseteq \probM$ is a nonnegative adapted stochastic process $(E_t)_{t \in \N}$ with
	$\E_\Q (E_\tau) \leq 1$ for all $\Q\in \H$ and all (possibly infinite) stopping times $\tau$.
	Equivalently, an adapted nonnegative process $(E_t)_{t \in \N}$ is an e-process for $\H$ if, for each $\Q \in \H$, it is upper bounded by a test supermartingale. In other words, the upper bounding test supermartingale may be different for each $\Q \in \H$.
	Hence, any test supermartingale is also an e-process. Importantly, Ville's inequality continues to hold for e-processes \citep{Ramdas2020}. 
	
	When testing more than one hypothesis simultaneously with e-processes, multiple testing corrections are necessary. We discuss the relevant background in Supplement \ref{app:pSMCS}.

	\subsection{SMCSs for the strongly superior and uniformly weakly superior objects}\label{Subsec:SMCS_by_p_processes}

	In this section, we construct SMCSs for the strongly superior models $\ssupM$ given at \eqref{eq:def_strongly_superior_objects} and for the uniformly weakly superior models $\uwsupM$ given at \eqref{eq:def_uniformly_weakly_superior_objects} using e-processes and sequential testing procedures.
	
	We define the strong hypotheses
	$\H^{\textrm{s}}_{ij}= \{\Q \in \probM \mid \mu_{ij,r} \leq 0, \forall t \in \N\}$, and the weak hypotheses $\H^{\textrm{uw}}_{ij}= \{\Q \in \probM \mid \Delta_{ij,r} \leq 0, \forall t \in \N\}$ for $i,j\in\M_0$. We would like to test all these pairwise hypothesis simultaneously at each time point $t \in \N$.
	For $\bullet \in \{\mathrm{s},\mathrm{uw}\}$, assume that, for any $i,j\in\M_0$, $(E_{ij,t})_{t \in \N}$ is an e-process for the hypothesis $\H^\bullet_{ij}$.
	Then, for any $i\in \M_0$, the arithmetic mean $E_{i\cdot,t} =  1/(m-1)\sum_{j \neq i} E_{ij,t}$ is an e-process for the intersection hypothesis $\H^{\bullet}_{i\cdot} = \cap_{j\neq i} \H^\bullet_{ij}$. We adjust the e-processes $( E_{1\cdot, t})_{t\in \N}, \dots, ( E_{m\cdot,t})_{t\in \N}$ for multiple testing using the closure principle \citep{Markus_etal_1976} with the arithmetic mean as $e$-merging function, that is, 
	\begin{equation}\label{eq:adjusted_e-processes}
		E^\star_{i\cdot,t}=\min_{I \subseteq \{1, \dots, m\}: i\in I} \frac{1}{\lvert I \rvert} \sum_{j \in I} E_{j\cdot,t}\leq E_{i\cdot,t}, \quad i\in \M_0, t \in \N,
	\end{equation}
	see Supplement \ref{app:pSMCS} for further details.
	Importantly, each adjusted process \((E_{i\cdot, t}^\star)_{t\in \N}\) is an e-process for the hypothesis \(\H_{i\cdot}^{\bullet}\).
	As shown in \citet[Proposition 3.1.]{Vovk_Wang_2021}, the arithmetic mean essentially dominates all other symmetric e-merging functions.
	Algorithm 1 of \citet{Vovk_Wang_2021} allows for a computation of the minimum at \eqref{eq:adjusted_e-processes} with computation time \(O(m^2)\). In fact, the computation time may be improved to $O(m \log m)$, see Supplement \ref{app:algo1} for details.
	Finally, for some significance level $\alpha\in (0,1)$, we define
	\begin{equation}\label{eq:def_SMCS_strong_hypthesis_e_process}
		\widehat{\M}_t= \left\{i \in \M_0 \mid  E_{i\cdot,t}^{\star} <  1/\alpha
		\right \}, \quad t\in \N.
	\end{equation}
	In words, at each $t\in \N$, we exclude model $i\in \M_0$ if we may reject the hypothesis $\H^{\bullet}_{i\cdot}$ according to the sequential test $\psi_{i\cdot,t}^\star = \one\{E_{i\cdot,t}^{\star} \ge 1/\alpha\}$, and include it otherwise.
	The proof the the following theorem is given in Supplement \ref{app:pSMCS}.
	
	\begin{thm}\label{thm_SMCS_by_e-processes}
		For any $\alpha \in (0,1)$, the sequence $(\widehat{\M}_t)_{t \in \N}$ defined at \eqref{eq:def_SMCS_strong_hypthesis_e_process} is an SMCS at level $\alpha$ for $\M^{\bullet, \star}$, $\bullet \in \{\mathrm{s},\mathrm{uw}\}$, and so is its running intersection $ \widetilde{\M}_t= \bigcap_{r\leq t} \widehat{\M}_r, t\in \N$. 
	\end{thm}
	
	Theorem \ref{thm_SMCS_by_e-processes} works since we control the family-wise error rate of the tests $\psi_{i\cdot,t}^\star$ over all models $i$. Family-wise error rate control relies on the closure principle and some e-merging function, see Supplement \ref{app:pSMCS}. Alternatively, one could adjust the pairwise e-processes $E_{ij,t}$ directly to obtain adjusted e-processes for the pairwise hypotheses $ \H^{\bullet}_{ij}$ and exclude the model $i$ from the SMCS at $t\in \N$ if $E^{\star}_{ij,t}\geq1/\alpha$ for some $j\neq i$. However, the corresponding tests can be shown to be uniformly less powerful, see Supplement \ref{app:pSMCS}.

	\begin{rem} Except from the family-wise error rate, the false discovery rate (FDR) is arguably the most commonly used criterion in multiple testing \citep{BenjaminiHochberg1995,BenjaminiYekutieli2001}. Recently, \cite{Wang_Ramdas_FDR_control} studied FDR control for e-values. In principle, SMCSs could be constructed based on FDR control. However, this fundamentally changes the type of validity we are obtaining for $\widehat{\M}_t$. It allows us to bound the fraction of wrongly rejected models by the total number of rejected models at all time points, which is different from our target coverage guarantee at \eqref{eq:def_coverage}. We give some further details in Supplement \ref{appendix:FDR_control}. 
	\end{rem}
	
	\begin{rem}\label{rem_marginal_coverage}
		Our anytime-valid methods readily extend if one is willing to weaken the simultaneous guarantee given at \eqref{eq:def_coverage} by marginal coverage, as defined right before Section \ref{Subsec:Discussion_of_superior_models}, thereby allowing to omit the adjustment step for multiple testing. The same potential gain in power may be achieved under the additional assumption that there is only one single superior object. See Supplement \ref{appendix:marginal_coverage} for further details. 
	\end{rem}
	
	The following proposition provides e-processes for the strong hypotheses under the assumption of conditionally bounded loss differences.
	Together with Theorem \ref{thm_SMCS_by_e-processes}, it directly leads to the derivation SMCSs for the strongly superior objects.
	\begin{prop}\label{prop:3.2}
		Assume that $|d_{ij,t}|\leq c_{ij,t}/2$ for $i,j \in \M_0$, $t \in \N$, and some predictable sequence $(c_{ij,t})_{t\in \mathbb{N}} \subseteq (0,\infty)$. Then, for any $i,j\in \M_0$, $E_{ij,t} = 
		\prod_{r=1}^t (1+\lambda_{ij,r} d_{ij,r})$ is an e-process for $\H_{ij}^\textrm{s}$, for any predictable $(\lambda_{ij,t})_{t \in \N}$ with $0\leq \lambda_{ij,t} \leq 1/c_{ij,t}$. 
	\end{prop}
	\begin{proof} The process $ M_{ij,t} = 
		\prod_{r=1}^t (1+\lambda_{ij,r} (d_{ij,r}-\mu_{ij,r}))$ is a nonnegative martingale with $M_{ij,0}=1$. The claim follows by observing that, for $\Q \in \H^{\textrm{s}}_{ij}$, $M_{ij,t}\geq E_{ij,t}$ for all $t\in \N$.
	\end{proof}
	The following proposition is due to \citet[Theorem 3]{Choe_Ramdas_2021}, and allows to construct e-processes for the weak hypothesis under the assumption of uniformly bounded loss differences. Together with Theorem \ref{thm_SMCS_by_e-processes}, it provides SMCSs for the uniformly weakly superior objects. In Section \ref{Subsec:Predictable_bounds_instead_of_fixed_bounds}, we explain how we can weaken the assumption of uniformly bounded loss differences to conditionally bounded loss differences. 
	\begin{prop}\label{prop:e-process_choe_ramdas}
		Let $|d_{ij,t}|\leq c_{ij}/2$ for $i,j \in \M_0$, $t \in \N$, and some $c_{ij}>0$. Then, for any $i,j\in \M_0$, \begin{equation*}
			E_{ij,t} = 
			\exp\left\{\lambda_{ij} t\hat{\Delta}_{ij,t}-\psi_{E,c_{ij}}(\lambda_{ij})V_{ij,t}\right\}, \quad t\in \N,
		\end{equation*} is an e-process for $\H_{ij}^\textrm{uw}$, for any $0\leq \lambda_{ij} \leq 1/c_{ij}$ and
		\begin{equation}\label{eq:Psi_{E,c}}
			\psi_{E,c_{ij}}(\lambda)=(-\log(1-c_{ij}\lambda_{ij})-c_{ij}\lambda_{ij})/c_{ij}^2   \eqand
			V_{ij,t} = \sum_{r=1}^t (d_{ij,r}-\gamma_{ij,r})^2,
		\end{equation}
		where $(\gamma_{ij,t})_{t \in \N}$ with $\lvert \gamma_{ij,t}\lvert \leq c_{ij}/2$ is a predictable sequence.
	\end{prop}

	\subsection{SMCSs for the weakly superior objects}\label{Subsec:SMCS_for_weakly_superior_objects}
	
	Our SMCS construction for the weakly superior objects defined at \eqref{eq:def_weakly_superior_objects} builds on time-uniform confidence regions.
	Proofs for this section are given in Supplement \ref{Appendix:proofs}.
	
	For $i,j \in \M_0$, $t \in \N$, $x \in \R$, let $M_{ij,t}(x)$ be such that $(M_{ij,t}(\Delta_{ij,t}))_{t \in \N}$ is a test supermartingale for any $\Q \in \probM$. Let $\diagM$ be the family of all $m\times m$-matrices with diagonal entries equal to zero and define for $\bX=(x_{ij})_{ij} \in \diagM$,
	\begin{equation*}
		M_t(\bX) = \frac{1}{m(m-1)}\sum_{\substack{i,j=1, \dots, m \\ i\neq j}}M_{ij,t}(x_{ij}), \quad t\in \N. 
	\end{equation*}
	With $\bDelta_t = (\Delta_{ij,t})_{ij} \in \diagM$, also $(M_t(\bDelta_t))_{t \in \N}$ is a test supermartingale for any $\Q \in \probM$. 
	Therefore, for any $\alpha \in (0,1)$, $C_t = \{\bX \in \diagM \mid M_t(\bX) \le 1/\alpha\}$ is a $(1-\alpha)$-confidence sequence for $(\bDelta_t)_{t \in \N}$.  Indeed, for any $\Q \in \probM$, by Ville's inequality, 
	\[
	\Q(\exists t \in \N : \bDelta_{t} \notin C_t) = \Q(\exists t \in \N : M_{t}(\bDelta_{t}) > 1/\alpha) \leq \alpha.
	\]
	
	We construct an SMCS from $C_t$ as follows. If there exists a $j \not=i$ such that we may reject $\Delta_{ij,t} \le 0$, we exclude $i$ from the SMCS.
	This is the case, if and only if there exists $j \not=i$ such that $C_t \cap \{\bX \in \diagM \mid x_{ij} \le 0\} = \emptyset$. Therefore, we define
	\begin{equation}\label{eq:def_SMCS_by_confidence_regions}
		\widehat{\M}_t = \left\{i \in \M_0 \mid  C_t \cap \{\bX \in \diagM \mid x_{ij} \le 0\} \not= \emptyset \text{ for all }j\neq i\right\}, \quad t\in \N.
	\end{equation}
	
	\begin{thm}\label{thm:3.5}
		For any $\alpha \in (0,1)$, the sequence $(\widehat{\M}_t)_{t \in \N}$ defined at \eqref{eq:def_SMCS_by_confidence_regions} is an SMCS at level $\alpha$ for the weakly superior objects $(\wsupM_t)_{t\in \N}$. Its running intersection $ \widetilde{\M}_t= \bigcap_{r\leq t} \widehat{\M}_r, t\in \N$ is an SMCS at level $\alpha$ for the uniformly weakly superior objects $\uwsupM$. 
	\end{thm}
	
	Similarly to Remark \ref{rem_marginal_coverage}, we may construct more powerful SMCSs if we target a suitable marginal coverage guarantee, see Supplement \ref{appendix:marginal_coverage} for details. 
	
	For computation of the SMCS, it is useful to understand the shape of the sets $C_t$. The following result treats a relevant case that allows for simpler computations. 
	
	\begin{prop}\label{prop:convex_upper_set}
		Suppose that the functions $x \mapsto M_{ij,t}(x)$ are nonnegative, convex and decreasing for all $i,j \in \M_0$ and $t \in \N$.  Then, $C_t$ is a convex upper set, that is, for $\bA, \bB \in C_t$, it follows that $\lambda \bA + (1-\lambda) \bB \in C_t$ for all $\lambda \in [0,1]$, and $\bA \in C_t$ and $\bA \leq \bC$ implies that $\bC \in C_t$, where the inequality is understood componentwise. 
	\end{prop}
	
	Under the conditions of Proposition \ref{prop:convex_upper_set}, the definition of $\widehat{\M}_t$ at \eqref{eq:def_SMCS_by_confidence_regions} simplifies:
	\begin{equation}\label{eq:def_SMCS_by_confidence_regions_2}
		\widehat{\M}_t = \left\{i \in \M_0 \mid  C_t \cap \{\bX \in \diagM \mid x_{ij} = 0\} \not= \emptyset \text{ for all }j\neq i\right\}, \quad t\in \N.
	\end{equation}
	
	For uniformly bounded loss differences, the following proposition provides test supermartingales that satisfy the assumptions of Proposition \ref{prop:convex_upper_set}. It can be found in \citet[Proposition 1 of Appendix]{Choe_Ramdas_2021}. In the next section, we discuss how to accommodate conditionally bounded loss differences.
	\begin{prop}\label{prop:choe_ramdas_2}
		Let $|d_{ij,t}|\leq c_{ij}/2$ for $i,j \in \M_0$, $t \in \N$, and some $c_{ij}>0$. Then, for any $i,j\in \M_0$, $(M_{ij,t}(\Delta_{ij,t}))_{t \in \N}$ is a test supermartingale for any $\Q \in \probM$, where
		\begin{equation*}
			M_{ij,t}(x) = 
			\exp\left\{\lambda_{ij} t\hat{\Delta}_{ij,t}-\lambda_{ij}tx - \psi_{E,c_{ij}}(\lambda_{ij})V_{ij,t}\right\}, \quad t\in \N,
		\end{equation*} for $0\leq \lambda_{ij} \leq 1/c_{ij}$ and $\psi_{E,c_{ij}}$ and $V_{ij,t}$ are given at \eqref{eq:Psi_{E,c}}. 
	\end{prop}
	
	Since \(\mathcal{M}^{\text{uw},\star}=\bigcap_{t=1}^{\infty} \mathcal{M}_t^{\text{w},\star}\), one could alternatively use the running intersection of the SMCSs defined at \eqref{eq:def_SMCS_by_confidence_regions_2} to construct SMCSs for the uniformly weakly superior objects.
	However, the SMCSs defined at \eqref{eq:def_SMCS_by_confidence_regions_2} are much more expensive computationally.
	Therefore, we only compute SMCSs for uniformly weakly superior objects using Theorem \ref{thm_SMCS_by_e-processes}.
	
	\subsection{Predictable bounds and betting schemes}
	\label{Subsec:Predictable_bounds_instead_of_fixed_bounds}
	
	Our methods for the (uniformly) weakly superior objects require that the score differences are uniformly bounded over time, which is clearly more restrictive than the assumption of conditionally bounded score differences imposed for the strongly superior objects. However, one can always transform conditionally bounded score differences into uniformly bounded ones thereby modifying the underlying loss and corresponding superior objects. 
	
	Assume that $\lvert d_{ij,t}\rvert \leq c_{ij,t}/2$ for all $t\in \N$ for some predictable $(c_{ij,t})_{t\in \N}$. Then, the transformed loss differences $\Tilde{d}_{ij,t}=d_{ij,t}/c_{ij,t}, t\in \N$ are uniformly bounded with constant $c=1$ that is $\lvert \Tilde{d}_{ij,t} \rvert \leq 1/2$ for all $t\in \N$. However, if we use the transformed loss differences, we target (possibly) different sequences of superior objects $\tilde{\M}^{\textrm{uw},\star}= \{i\in \M_0 \mid  \tilde{\Delta}_{ij,t}\leq 0, \text{for all}\; j\in \M_0, t \in \N\}$ and $ 
	\tilde{\M}^{\textrm{w},\star}_t= \{i\in\M_0 \mid \tilde{\Delta}_{ij,t}\leq 0, \text{for all} \; j \in \M_0\}$ for $\tilde{\Delta}_{ij,t} = (1/t)\sum_{r=1}^t \E(\tilde{d}_{ij,r}\mid \F_{r-1})=(1/t)\sum_{r=1}^t \mu_{ij,r}/c_{ij,r}.$
	
	In Section \ref{Sec:Case_study}, we use the given transformation to convert conditionally bounded CRPS differences of wind gust forecasts into uniformly bounded loss differences. Computational feasibility of this approach is due to the forecasts being either parametric or with finite support, see Supplement \ref{appendix:conditionally_bounded_scores} for details.
	Importantly, for any proper scoring rule (consistent scoring function) $\myS$, the scaled function $\tilde{\myS}=\myS/c$ is proper (consistent) as well for any $c>0$, an observation which justifies the transformation from a  theoretical perspective.
	This property can be extended in the sequential setting, see Supplement \ref{appendix:forecast_evaluation}.
	With the given scaling, all observations have the same maximal impact on the ranking, an effect which seems appealing for many (but surely not for all) applications.

	\begin{rem} For important choices of loss functions such as the quadratic or logarithmic score, the corresponding loss differences are not conditionally bounded.
		However, if the loss differences emerge from (or may be bounded in the tails by) a parametric family of distributions, there are other well-studied e-processes available; see, e.g., \cite{Howard_Chernoff_bounds}.
		In Supplement \ref{Appendix:MSE_Simulation}, we provide a simulation example of mean forecasts with Gaussian errors. The resulting score differences are unbounded but still sub-exponential, which allows to use the e-processes of Proposition \ref{prop:e-process_choe_ramdas}. 
	\end{rem}

	We conclude with a comment on the choice of the parameters $\lambda_{ij}$. For the uniformly weak and weak hypothesis (Propositions \ref{prop:e-process_choe_ramdas} and \ref{prop:choe_ramdas_2}), we assume universal bounds $c_{ij}>0$ on the pairwise loss differences, and $\lambda_{ij}$ may be any fixed value in the interval $(0,c_{ij}^{-1})$.
	As a default, we suggest to choose $\lambda_{ij} = (2c_{ij})^{-1}$. In betting language, this corresponds to bet half of the accumulated evidence at each time step \citep{Shafer2021}. For the strong hypothesis (Proposition \ref{prop:3.2}), we can choose $\lambda_{ij,t}$ predictably to increase power.
	There are many possible betting schemes.
	In Section \ref{Subsec:Covid}, we propose a particular betting scheme which makes use of all previous observations and the assumption that the relative performance of two models does not change too quickly over time.
	Another option would be the \emph{method of mixtures} \citep{Robbins_1970}, where we integrate over all possible values of $\lambda_{ij,t}$ over a particular probability distribution instead of choosing one specific value.
	Mixtures of e-processes are again e-processes and may be expressed in closed form for some distributions \citep[Appendix A.3]{Howard_confidence_sequences}.
	The method of mixtures is one of the most widely-studied techniques for constructing uniform bounds which are shown to be unimprovable under certain conditions (\citealp{Robbins_Sigmnund_1970}, \citealp{Howard_confidence_sequences}).
	However, in accordance with \citet{waudby2023estimating}, for the strong hypothesis, our experience shows that choosing $\lambda_{ij,t}$ predictably is generally more promising than the methods of mixtures.
	Also for the weak hypotheses, where predictable betting strategies do not apply, we could not improve the power of our methods significantly by using mixtures instead of choosing $\lambda_{ij}$ fixed.
	For this reason, and to make computations more tractable, we suggest to work with the given e-processes directly.

	\section{Simulations}\label{Sec:Simulation}
	\subsection{Simulation 1}

	We consider the setup given in Section \ref{Subsec:Discussion_of_superior_models}. That is, we consider $n=1000$ realizations from a sequence $Y_t \sim \mathcal{N}(Y_{t-1},1)$ with $Y_0=0$, and assess $m=49$ forecasters which issue sequentially probabilistic predictions with respect to the CRPS. We write $Y_t=Y_{t-1}+Z_t$, for an i.i.d.\  standard normally distributed sequence $(Z_t)_{t\in \N}$, 
	and use the formulas provided by \cite{Grimit_et_al} to derive 
	\begin{align}\label{eq:loss_sim1}
		L_{i,t}
		= \sqrt{1+\delta_i} \left(\frac{Z_t-\varepsilon_i}{\sqrt{1+\delta_i}}\left(2 \Phi\left(\frac{Z_t-\varepsilon_i}{\sqrt{1+\delta_i}}\right)-1\right) + 2 \varphi \left(\frac{Z_t-\varepsilon_i}{\sqrt{1+\delta_i}}\right)-\frac{1}{\sqrt{\pi}}\right),
	\end{align}
	where $\Phi$ and $\varphi$ denote the cdf and density function of the standard normal distribution, respectively. Let $(\F_t)_{t\in \N}$ be the canonical filtration, that is, $\F_t$ is generated by  $Y_r$ and  $f_{i,r}$ for $i=1, \dots, m$, $r\leq t$. For any $t\in \N$, we have $\E(d_{ij,t}\mid \F_{t-1})=\E(d_{ij,t})=\E(d_{ij,1})$ by independence. Since the $\crps$ is a proper scoring rule, we conclude $\ssupM= \left\{i \in \M_0 \mid \E(d_{ij,1})\leq 0, \forall j\in \M_0\right\}= \{i_0\}$ for the forecaster $i_0$ with $\varepsilon_{i_0}=\delta_{i_0} = 0$. 
	
	We construct SMCSs for $\ssupM$  with confidence level $\alpha=0.1$. Upper bounds on the score differences $(d_{ij,t})_{t\in \N}$ are given in Supplement \ref{appendix:conditionally_bounded_scores}. These do not depend on time, since biases and dispersion errors do not change, so $c_{ij,t}=c_{ij}>0$, $t\in \N$. Define $\lambda_{ij,t}=(2c_{ij})^{-1}$. 
	
	Figure \ref{fig:Simulations1&2_results} shows the average size of the SMCS over time for $N=1000$ simulations. We have a 100\% coverage rate for the best model $i_0$, that is, in all simulation runs, the SMCS never incorrectly excludes $i_0$. Hence, our methods are conservative. However, using e-processes and a sequential application of the closure principle, we cannot expect to improve the power of our methods for SMCSs, since the arithmetic average, which is the e-merging function that we used, essentially dominates any other e-merging function.

	\begin{figure}[ht]
		\centering
		\includegraphics[width =0.7\textwidth]{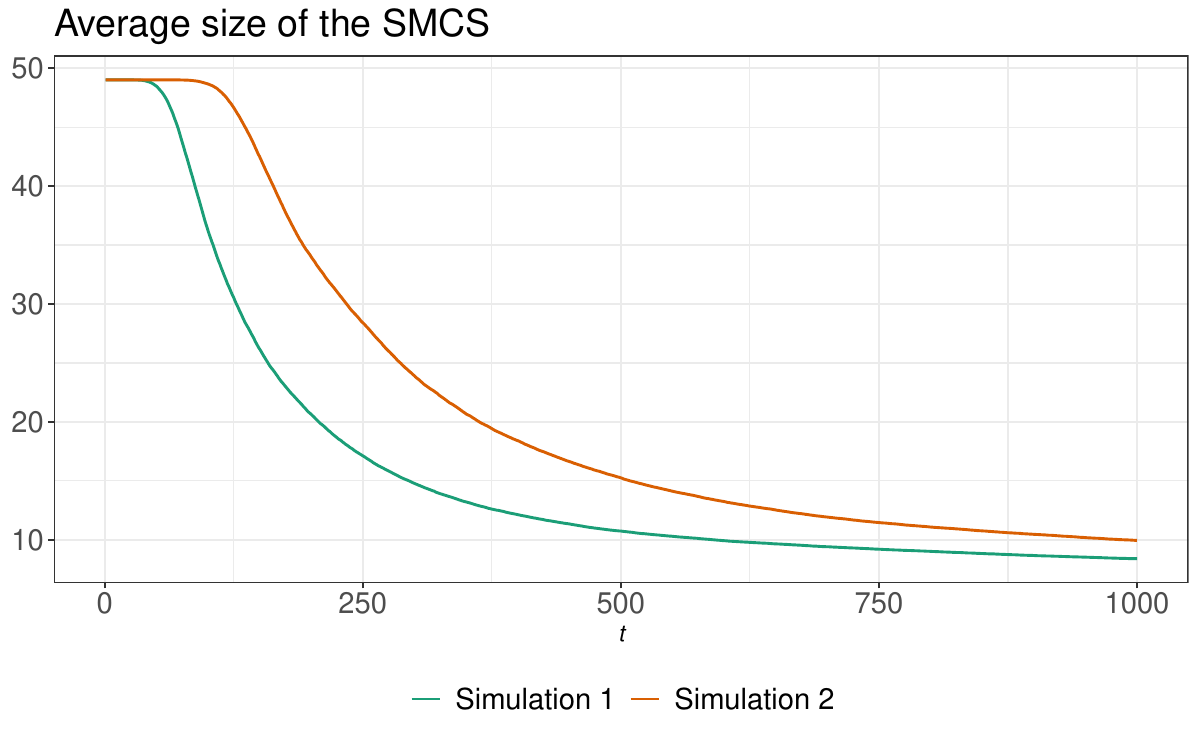}
		\caption{The average number of models in the SMCS in Simulation 1 and 2. At the end of the evaluation period, the SMCSs have an average size of $8.41$ and $9.95$, respectively. For both simulations, the SMCS never wrongly excludes the best model $i_0$.}
		\label{fig:Simulations1&2_results}
	\end{figure}
	
	We ran Simulation 1 also for the original MCS by \cite{MCS}.
	Although its size is generally better than that of the SMCSs, the coverage rate drops below the nominal level \(1-\alpha = 0.9\) to about $0.8$.
	The sequential implementation of the MCS is documented in Supplement \ref{app:F.2}, along with another simulation in a more challenging, non-Gaussian setting.
	There, the coverage rates of the MCS are heavily compromised, in contrast to those of the SMCS.

	\subsection{Simulation 2}
	
	We consider the simulation setting as given in Section \ref{Subsec:Discussion_of_superior_models} for $n=1000$. By \eqref{eq:loss_sim1},
	\begin{align}
		\E(L_{i,t})  &=  \E\left(Z_t \Phi\left(\frac{Z_t-\varepsilon_i}{\sqrt{1+\delta_i}}\right)\right) -\varepsilon_i \E \Phi\left(\frac{Z_t-\varepsilon_i}{\sqrt{1+\delta_i}}\right)+\varepsilon_i+2\E\varphi\left(\frac{Z_t-\varepsilon_i}{\sqrt{1+\delta_i}}\right)-\sqrt{\frac{1+\delta_i}{\pi}} \nonumber \\
		&= \frac{1+ 2 \sqrt{1+\delta_i}}{\sqrt{2\pi(2+\delta_i)}}\exp\left\{-\frac{\alpha_i^2}{2}\right\}-\varepsilon_i \Phi\left(-\alpha_i\right)+\varepsilon_i - \sqrt{\frac{1+\delta_i}{\pi}}.\label{eq:loss_Simulation2}
	\end{align}
	For $i,j \in \M_0$ and $t\in \N$, let $\E(d_{ij,t} \mid \F_{t-1})=\E(d_{ij,t})=\mu_{ij,t}$. Since the CRPS is a proper scoring rule, we have $\mu_{i_0 j,t}\leq 0$ for all $t\notin 7\N$ and $j\in \M_0$. By \eqref{eq:loss_Simulation2}, for $j\in \M_0$ with $(\varepsilon_j,\delta_j)\in \{-0.2,0.2\}^2$, we have $\mu_{i_0 j,t}> 0$ for all $t\in 7\N$. Thus, $\ssupM = \emptyset$. However, equation \eqref{eq:loss_Simulation2} allows to show numerically that $\Delta_{i_0j,t}\leq 0$ for all $j\in \M_0$ and $t\in \N$, and thus $\uwsupM = \{i_0\}$. 
	
	We compute SMCSs for $\uwsupM$ with confidence level $\alpha=0.1$ using the approach in Section \ref{Subsec:SMCS_by_p_processes}. We transform the conditionally bounded CRPS differences (see Supplement  \ref{appendix:conditionally_bounded_scores}) as proposed in Section \ref{Subsec:Predictable_bounds_instead_of_fixed_bounds} into uniformly bounded score differences $\lvert \tilde{d}_{ij,t} \rvert \leq c/2$ for $c=2$, and let $\lambda_{ij}=(2c)^{-1}=1/4$. 
	Figure \ref{fig:Simulations1&2_results} displays the average SMCS size across time, which declines slightly slower than in Simulation 1.
	This is to be expected, since we work with a weaker hypothesis.
	Regarding coverage of the superior object, the SMCS is still conservative and exhibits a coverage rate of \(100\%\).
	
	It could be interesting to monitor the e-processes that we use to collect evidence on predictive performance over time. Due to the large number of models considered here, this is not so practical, but we illustrate this idea in Supplement \ref{Appendix:MSE_Simulation}, where only nine models are considered in the experiment.
	
	In Supplement \ref{appendix:marginal_coverage}, we illustrate the potential gain in power by replacing simultaneous coverage by marginal coverage (see Remark \ref{rem_marginal_coverage}) by comparing the average size of the SMCS in Simulations 1 and 2 with and without multiple testing correction.
	
	\subsection{Simulation 3}
	
	We consider the simulation setting from Section \ref{Subsec:Discussion_of_superior_models}. Let $i=1,2,3$ be forecasters issuing median predictions $m_{i,t}=Y_t+\varepsilon_{i,t}$, $t\in \N$, for i.i.d.\ standard normally distributed $(Y_t)_{t\in \N}$ and $\varepsilon_{1,t}= \beta,  \varepsilon_{2,t}= \gamma^t, \varepsilon_{3,t}= \delta t,$ for $\beta, \delta >0$ and $0<\gamma<1$. For the loss $L(m,y)=0.5(\Phi(m)-\Phi(y))$, we get by independence
	that $2 \mu_{ij,t}=\E(\Phi( m_{i,t})) -\E(\Phi( m_{j,t})) = \Phi\left(\varepsilon_{i,t}/\sqrt{2}\right) - \Phi \left(\varepsilon_{j,t}/\sqrt{2}\right)$ for  $t\in \N$, $i,j=1,2,3$, thus
	\begin{equation}\label{eq:superior_models_sim3}
		\wsupM_t = \left\{i \ \middle| \  \sum_{r=1}^t \Phi\left(\frac{\varepsilon_{i,t}}{\sqrt{2}}\right) \leq \sum_{r=1}^t \Phi\left(\frac{\varepsilon_{j,t}}{\sqrt{2}}\right),\; \text{for all}\; j\right\}.
	\end{equation}
	We let $\beta=0.6, \gamma = 0.998$ and $\delta = 0.008$. Then, one can numerically show that $\wsupM_t = \{3\}$ for $t\leq 153$, $\wsupM_t = \{1\}$ for $153 < t< 550$, and $\wsupM_t = \{2\}$ for $t\geq 550$.
	
	We sample $n=800$ observations and construct the SMCSs defined in Section \ref{Subsec:SMCS_for_weakly_superior_objects} for $(\wsupM_t)_{t\leq n}$, at level $\alpha=0.1$.
	We use the fact that $(d_{ij,t})_{t\in \N}$ is uniformly bounded with $c=1$ and let $\lambda_{ij}=1.1^{-1}$, $i,j=1,2,3$. Figure \ref{fig:Simulation3_results} shows the accumulated observed losses $\sum_{r=1}^t L_{i,r}$ and the resulting SMCS for one realization. The SMCS correctly excludes the improving forecaster at the beginning and the worsening forecaster relatively quickly after the first change point ($t=154$) until it includes the improving forecaster again and finally excludes the constantly biased forecaster.
	The average size of the SMCS for $N=100$ runs varies between one and two models after a short initial period with three models. The SMCS includes the superior objects given at \eqref{eq:superior_models_sim3} at all time points. 
	
	Again, we show in Supplement \ref{appendix:marginal_coverage} how one can reduce the average size of the SMCS under a weaker coverage guarantee. 
	
	\begin{figure}[h]
		\centering
		\includegraphics[width =0.49 \textwidth]{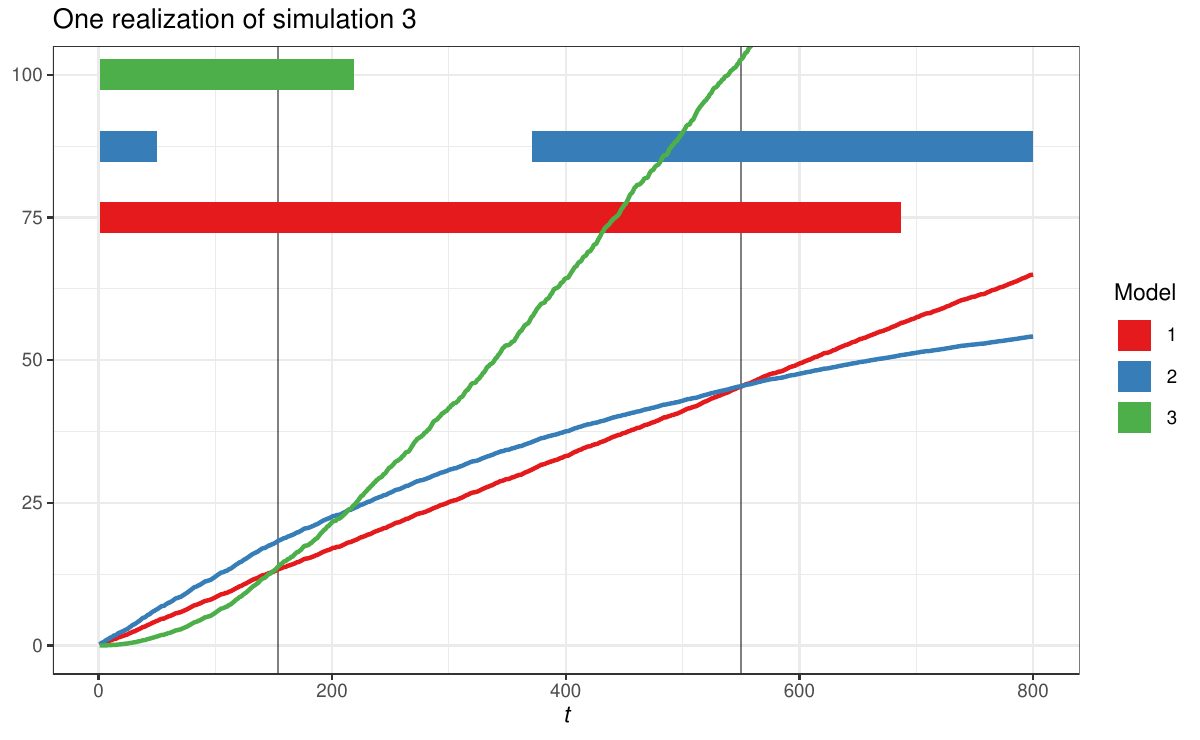}
		\includegraphics[width =0.49 \textwidth]{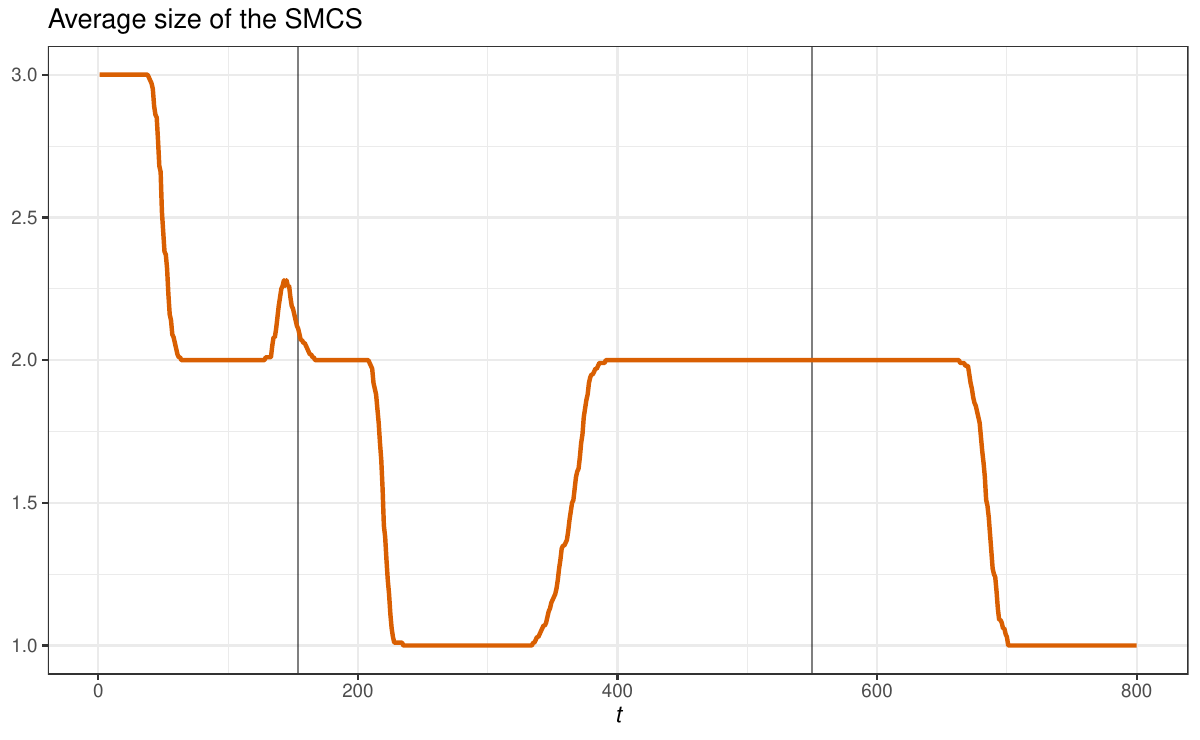}
		\caption{Left: Realized accumulated losses $\sum_{r=1}^t L_{i,r}$ for one realization: Worsening forecaster (green), improving forecaster (blue), constantly biased forecaster (red). The black vertical lines indicate $t=154$ and $t=550$. The resulting SMCS is given in the upper part with the respective colors. Right: Average size of the SMCS over $N=100$ realizations.}
		\label{fig:Simulation3_results}
	\end{figure}

	\section{Case Studies}\label{Sec:Case_study}
	
	\subsection{Covid-19 case study}\label{Subsec:Covid}
	
	After the outbreak of the Covid-19 pandemic, the Reich lab at the University of Massachusetts collaborated with the United States centers for disease control and prevention (CDC) to create the Covid-19 forecast hub, a repository containing point and probabilistic forecasts for incident cases, hospitalizations, and Covid-19 related deaths. The repository was founded in March 2020, and attracted submissions of forecasts from more than \(60\) different teams.
	The collection of forecasts is still ongoing, although most of the models had stopped submitting forecasts by January 2024, due to the subsidence of the pandemic.

	\subsubsection{Data}
	
	The dataset we use in this application, is 
	discussed in detail in \cite{Cramer2021-hub-dataset,Cramer2022_evaluation_of_forecasts}, and 
	publicly available at the Covid-19 forecast hub GitHub repository and on the Zoltar forecast archive. 
	It contains forecasts with corresponding observations for key epidemiology indices for 55 different locations in the U.S.~as well as on the aggregated U.S.~national level. The forecasts are issued by a total of 69 different forecast models and reported at \(23\) different quantile levels with forecast horizons ranging from $1$ day to $20$ weeks. Some models are not statistically comparable since they have consistently issued their forecasts on different weekdays.
	There are also models that started issuing forecasts late into the forecasting period and others that stop early.
	Finally, there are models with a large number of missing values.
	In principle, one could use missing value imputation, or assign a default score to missing values so that the number of comparable models is higher.
	However, the purpose of this study to demonstrate how sequential model confidence sets work, so we so not go further in this direction.
	We only include models that are directly comparable without any  preprocessing.
	
	Specifically, we consider 1-week-ahead forecasts for Covid related deaths on the national level, issued weekly in the period from 06/06/2020 to 04/03/2024, and focus on the comparison of the following $m=6$ different forecasting models:
	Firstly, we consider the benchmark model \emph{baseline}, which naively issues the most recent outcome as the median prediction for the following week and forms a predictive distribution around this median prediction based on the past weekly incidences. Secondly, we consider two epidemic models \emph{PSI-draft} and \emph{MOBS-gleam} issuing their forecasts based on epidemiological assumptions, and the neural network based model \emph{GT-deep}, which was the first purely data-driven model to be included in the Covid-19 forecast hub. Finally, we consider the summary models \emph{ensemble} and \emph{CDC-ensemble}, where the former combines the most recent predictions of all other models into one predictive distribution, and the latter is supposed to improve the former by only considering the 10 current best models, measured by the average weighted interval score over the 12 most recent weeks. For a more detailed documentation of the models, we refer to \cite{Cramer2021-hub-dataset}.

	We perform our comparison under the strong hypothesis, since we expect the ensemble models to have consistent superior performance over the individual models, which in turn are expected to be better than the naive baseline model issuing its forecasts only based on the previous outcomes and on no additional information. Furthermore, we can argue that in such a highly important prediction task, we should seek to detect the models that issue accurate predictions consistently over time and not just on average. 
	
	The outcomes as well as the forecasts exhibit unstable behaviour and  large variation in scale, see e.g.\ the upper left panel of Figure \ref{fig:covid_data}, which suggests to apply a log transformation on all quantities of interest. Moreover, forecast differences are typically not distinctive for the median, while they become more apparent when we look at the tails, see the lower panels in Figure \ref{fig:covid_data}.
	\begin{figure}[t]
		\centering
		\includegraphics[width=0.49\textwidth]{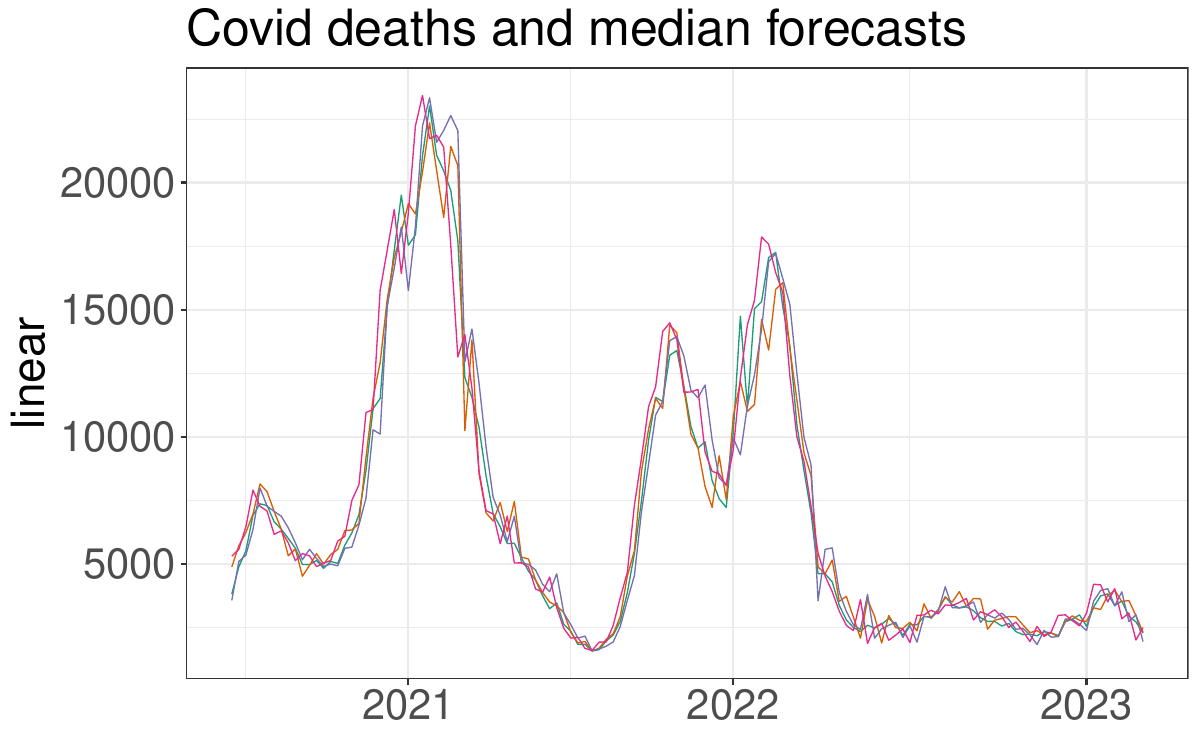}
		\includegraphics[width=0.49\textwidth]{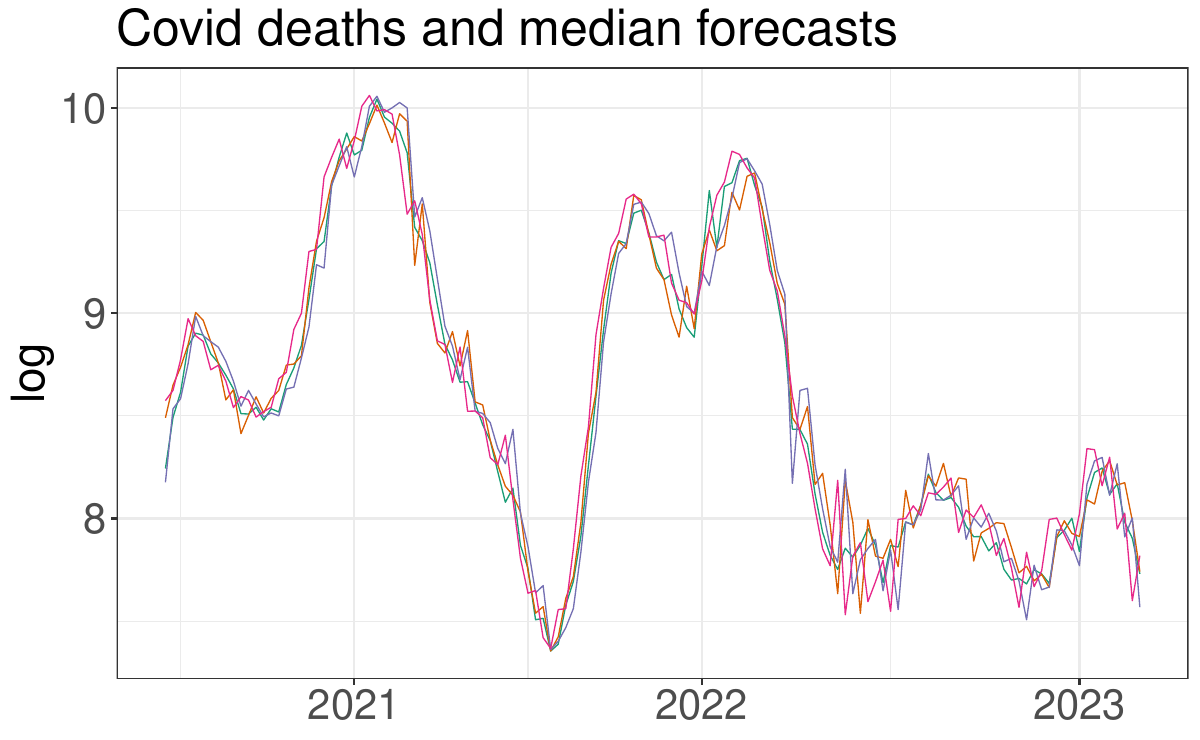}
		\includegraphics[width=0.49\textwidth]{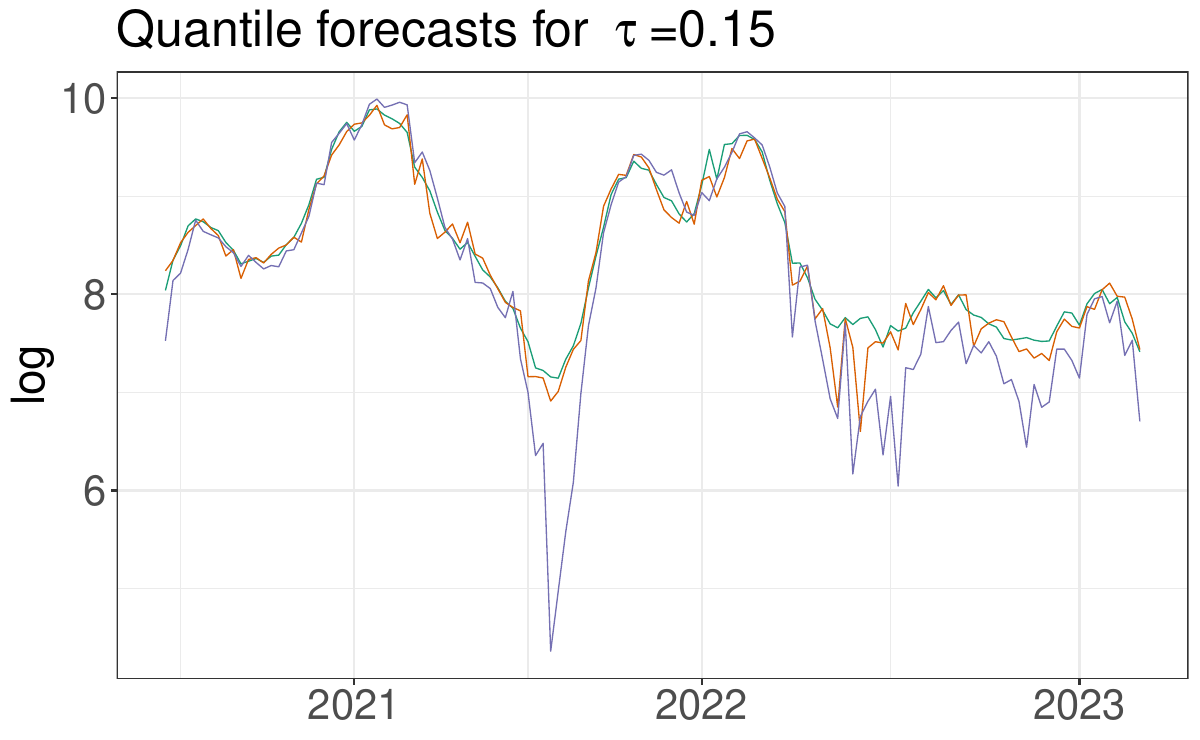}
		\includegraphics[width=0.49\textwidth]{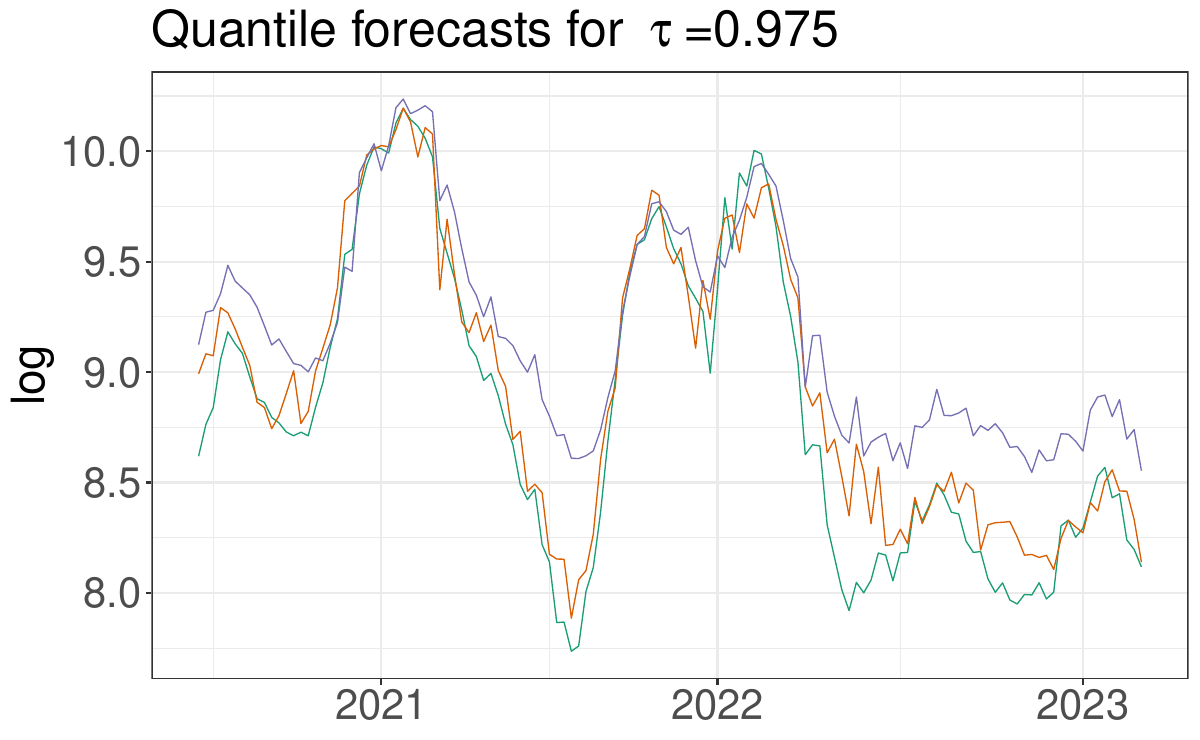}
		\includegraphics[width=0.6\textwidth]{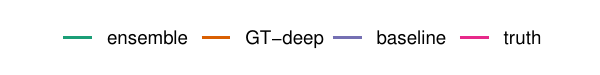}
		\caption{Time progression of predicted and actual Covid related mortality in linear (upper left) and log scale (upper right). For the forecasts of the tails, the differences between the different models are more pronounced than for the median, see the lower panels for \(\tau=0.15,0.975\).}
		\label{fig:covid_data}
	\end{figure}

	\subsubsection{Methods and implementation details}
	
	We assess the different quantile forecasts at level $\tau \in (0,1)$ with respect to the generalized piecewise-linear quantile scoring function $S_\tau(x,y)= \left(\mathbf{1}\{x\geq y\}-\tau\right)\left(\log(x)-\log(y)\right)$ for $x,y\in \R$, see \citet{Gneiting2011a} and Supplement \ref{appendix:forecast_evaluation}.
	In Supplement \ref{appendix:conditionally_bounded_scores}, it is shown that the resulting quantile score differences are conditionally bounded. More precisely, for $\tau$-quantile predictions $(x_{i,t})_{t \in \N}, (x_{j,t})_{t \in \N}\subseteq \R$ issued by the models $i,j=1, \dots, m$, we have $\lvert d_{ij,t} \rvert =\lvert S_\tau(x_{i,t},y_t)-S_\tau(x_{j,t},y_t) \rvert \leq c_{ij,t}/2$ for
	\begin{equation}\label{eq:predictable_bounds_covid_case_study}
		c_{ij,t} = 2\max\left\{\tau, 1-\tau \right\}\lvert \log(x_{1,t})-\log(x_{2,t})\rvert, \quad t\in \N.
	\end{equation}
	We apply our methods from Section \ref{Subsec:SMCS_by_p_processes} to construct SMCSs for the strongly superior objects, using the e-processes
	$E_{ij,t} = \prod_{r=1}^t (1+\lambda_{ij,r} d_{ij,r})$ for a predictable sequence $(\lambda_{ij,t})_{t\in \N}$ with $0<\lambda_{ij,t} < c_{ij,t}^{-1}$, $t\in \N$. For $i,j=1, \dots, m$ and $t\in \N$, we suggest $\lambda_{ij,t}=({K_{ij,t}c_{ij,t}+\varepsilon})^{-1}$ for some small \(\varepsilon>0\), included as a safeguard for the scenario \(c_{ij,t}=0\), which would be the case if the two forecasters issued the same prediction, and
	\begin{equation*}
		K_{ij,t} = \frac{2-\left|\tau-\frac{1}{2}\right|}{1+\left|\tau-\frac{1}{2}\right|} \cdot \frac{3\pi/2+\arctan(-d_{ij,t-1})}{\pi} \geq 1.
	\end{equation*}
	The coefficient \(K_{ij,t}\) depends on the sign and magnitude of the most recent score difference \(d_{ij,t-1}\) as well as of the centrality of $\tau$: If \(d_{ij,t-1}>0\), then we have evidence against the null hypothesis that model $i$ is strictly superior than model $j$, and the quantity \(K_{ij,t}\) becomes small. 
	That is, the parameter \(\lambda_{ij,t}\) tends to be larger, and we bet more against the null hypothesis. 
	Finally, the primary multiplicative factor, which takes values in $[1,2]$, is used to reduce the influence of the term \(\max\left\{\tau, 1-\tau\right\}\) on the bound in \eqref{eq:predictable_bounds_covid_case_study}. 
	
	\subsubsection{Results}
	
	\begin{figure}[t]
		\centering
		\includegraphics[width=0.49\textwidth]{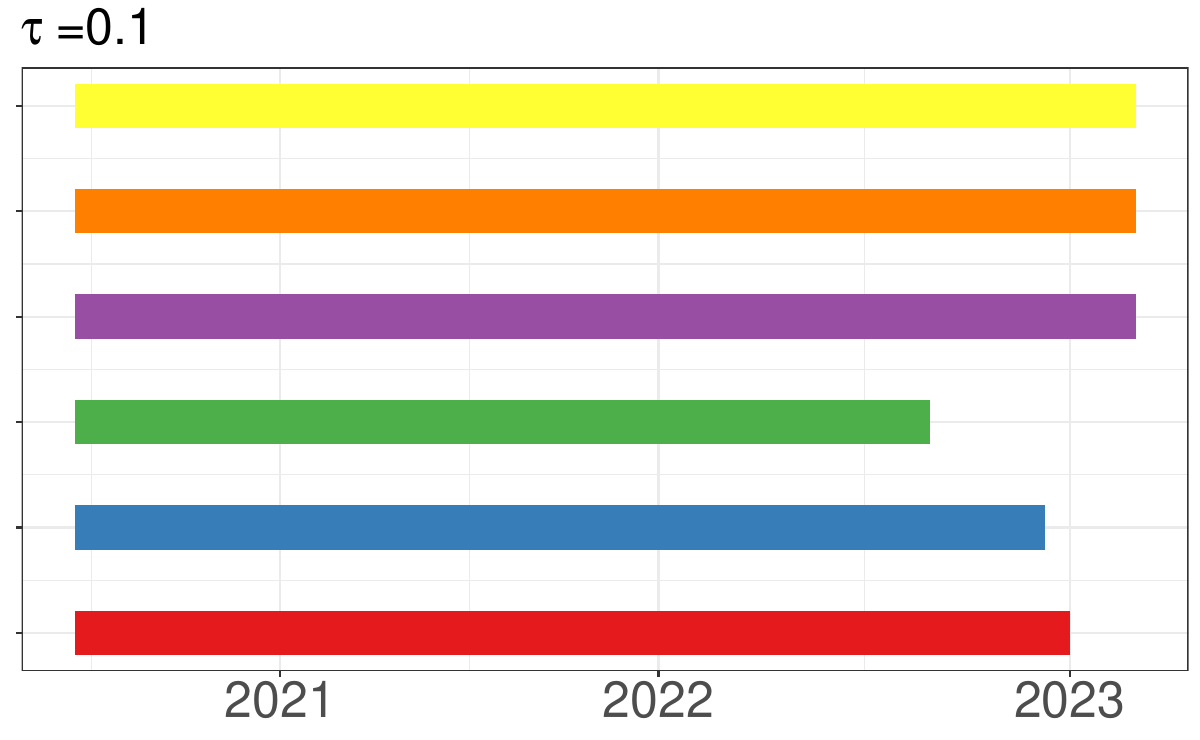}
		\includegraphics[width=0.49\textwidth]{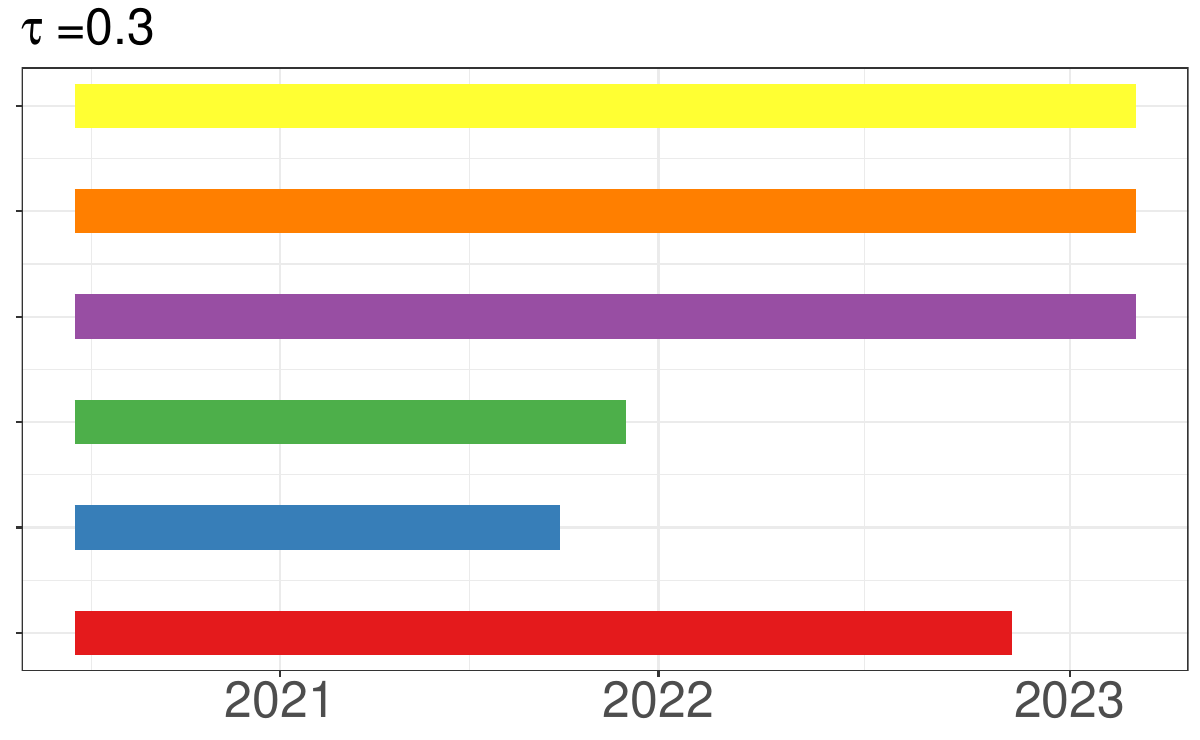}
		\includegraphics[width=0.49\textwidth]{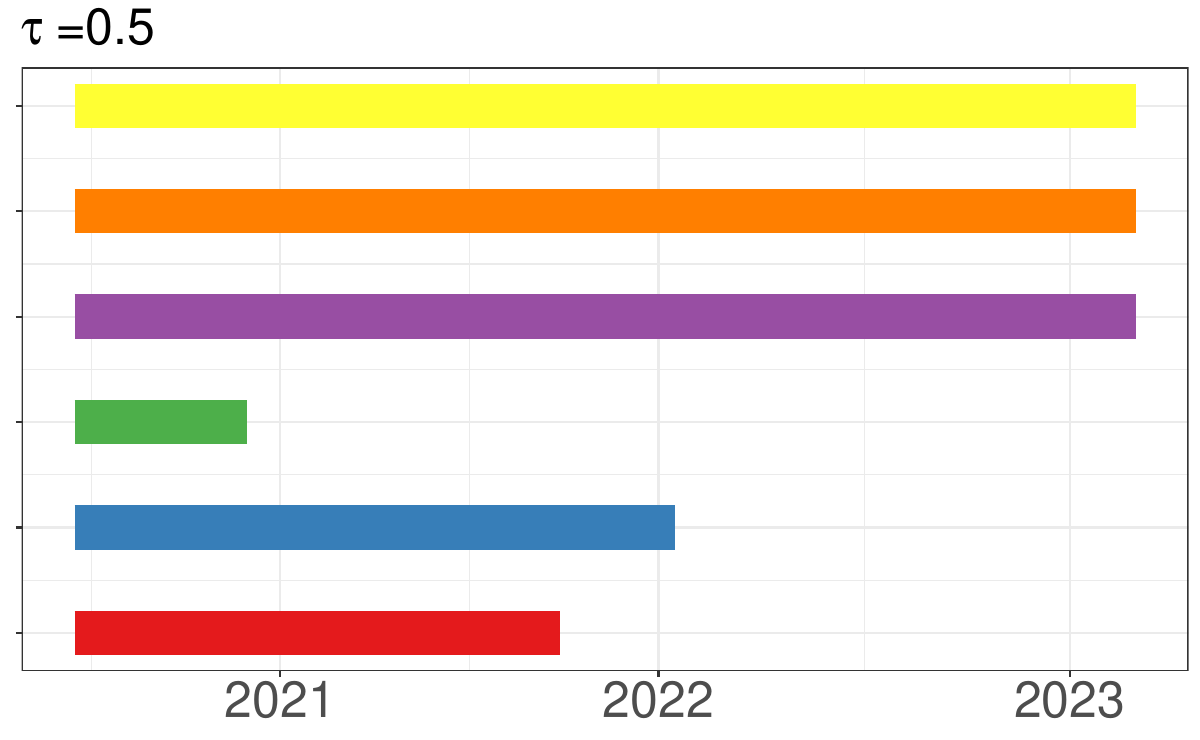}
		\includegraphics[width=0.49\textwidth]{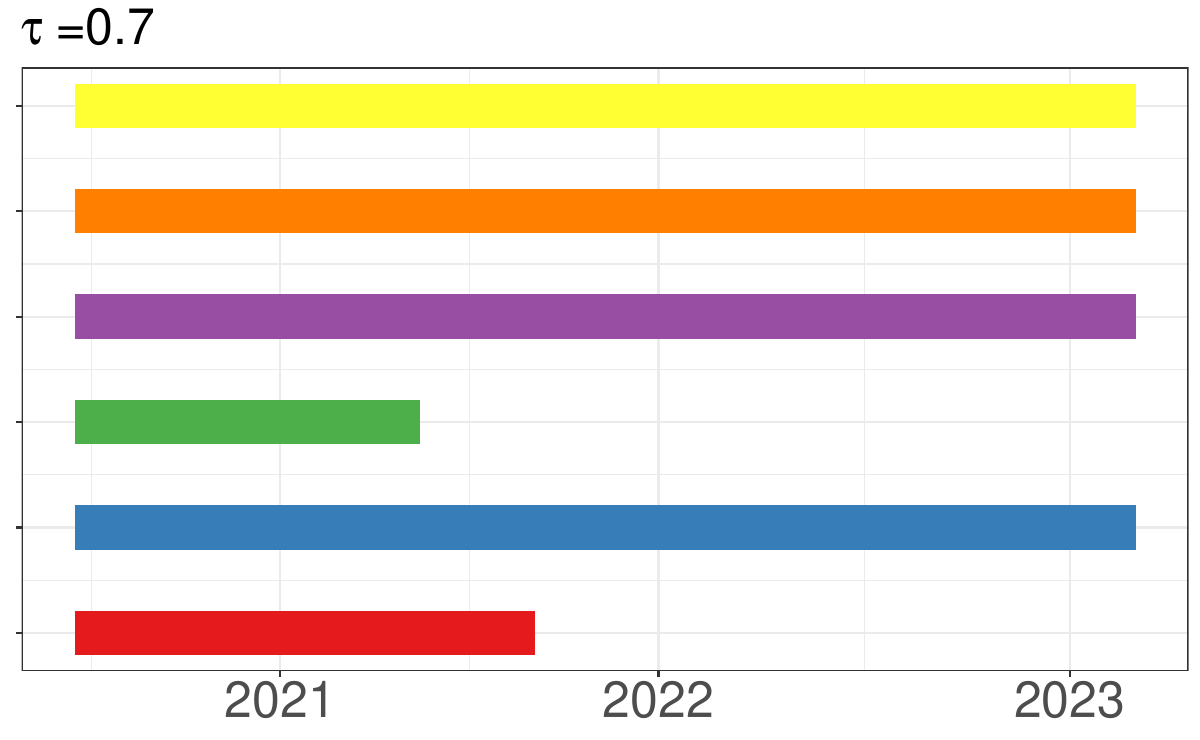}
		\includegraphics[width=0.8\textwidth]{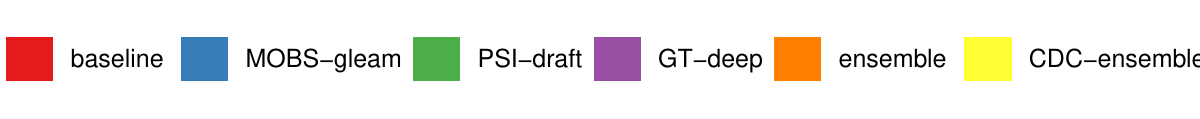}
		\caption{SMCSs at four selected quantile levels with confidence level $\alpha=0.1$.}
		\label{fig:results_covid_smart_approach}
	\end{figure}
	
	Figure \ref{fig:results_covid_smart_approach} shows the resulting SMCSs for four selected quantile levels with a confidence level of $\alpha=0.1$. As a first observation, we see that the power of the SMCS varies with the choice of the quantile level $\tau$: Whereas the SMCS detects three deficient models for $\tau=0.1,0.3,0.5$, it only has excluded two models by the end of the observation period for $\tau=0.7$. For $\tau=0.1$, the SMCS needs considerably more observations to reject the three models, which are excluded much faster at the median. The lower power in the tails of the distribution may be due to the factor \(\max\left\{\tau, 1-\tau\right\}\) in the bound \eqref{eq:predictable_bounds_covid_case_study}. 
	Overall, the results show that the two ensemble models perform best, as they are included in the SMCS across the entire time period, and at all quantile levels.
	The superior performance of the two ensemble models was expected, since they combine the forecasts of all other models and thus have access to a larger information set. Interestingly, the SMCS does not detect any significant difference between ensemble and CDC-ensemble, where the latter was supposed to improve the former by combining only the best-performing forecasting models.
	The naive baseline performs surprisingly well and even outperforms the epidemiological model, PSI-draft, which is excluded earlier than the baseline for all given values of $\tau$.
	We may conclude that PSI-draft is not competitive against the other models and should therefore not be considered as an accurate forecasting model for the pandemic. Finally, the second epidemiological model, MOBS-gleam, performs almost always better than PSI-draft but not as good as the neural network based GT-deep.
	
	To conclude, we highlight once again the key feature that SMCSs allow us to collect statistical evidence sequentially over time.
	Hence, if we were back in the pandemic, 
	by using SMCSs, we could monitor the models performances from the beginning on and would not have to wait until we have collected enough observations to perform a statistical test for predictive ability. Quickly, we would reject PSI-draft as an accurate model and not consider it in our analysis and policy-making anymore. Also MOBS-gleam would probably have been rejected as a competitive model by approximately the end of 2021.
	Clearly, in such emergency situations, it has already been common practice to assess the forecasting models on a regular basis. However, this practice may lack theoretical justification, and our methodology provides statistically safe methods to do so. 
	
	We focused primarily on the sequential monitoring of model performances but our methodology also allows for optional stopping. For example, a possible stopping criterion back in the pandemic could have been the first day when there are only half of the original models left.
	
	Throughout the paper we have used a significance level equal to $\alpha = 0.1$, but clearly there is freedom to do otherwise. The choice of $\alpha$ affects the interpretation of the SMCS. For example, when $\alpha=0.1$, we can deduce that the SMCS contains the optimal Covid-19 forecasting model with  90\% confidence. 
	
	The choice of $\alpha$ also affects the size of the SMCS. Larger values of $\alpha$ would lead to earlier detections and to a smaller SMCS.
	However, they would also translate to a lower level of confidence.
	In practice, we should keep this tradeoff in mind when choosing \(\alpha\).
	Our goal to narrow down the original set of models to a small SMCS competes with the high level of confidence that is necessary in various applications.
	
	\subsection{Postprocessing of wind gust predictions}\label{Subsec:Windgust}
	
	Weather services use numerical weather prediction (NWP) models for predicting the future state of the atmosphere. These NWP models quantify forecast uncertainty using ensemble predictions, where each ensemble member represents a different future scenario. However, such ensemble forecasts are typically subject to systematic biases and dispersion errors. Hence, the forecasts ought to be statistically postprocessed to generate accurate and reliable predictions \citep{Gneiting2005a, Vannitsem2018}. 
	
	\citet{schulz2022machine} present a comprehensive comparison of multiple approaches for statistical postprocessing of ensemble forecasts for wind gusts. 
	Here, we will build on their case study and compare the different forecasting models for probabilistic wind gust prediction by using SMCSs under the weak hypothesis.

	\subsubsection{Data}
	
	We use the data set in \citet{primo2024comparison}, which is an extension of the original data in \citet{schulz2022machine}. 
	The forecast data is based on an operational NWP model from the German weather service (Deutscher Wetterdienst, DWD) in the period from 08/12/2010 to 30/06/2023 with corresponding observations at 174 geographically diverse weather stations in Germany. The forecasts are initialized at 00 UTC and have a forecast horizon of up to 21 hours. Here, we consider only those forecasts with a lead time of 18h referring to 18 UTC and the $K=166$ stations with at least 2,600 observations.
	
	Next to the ensemble forecasts (EPS), we include the eight postprocessing methods applied by \citet{schulz2022machine}. 
	Statistical postprocessing methods are typically distributional regression models that use ensemble predictions from the NWP model as input data.
	The methods in we compare can be divided into three groups: Basic techniques that only make use of the underlying wind gust ensemble predictions (EMOS, MBM, IDR), established machine learning methods for postprocessing that incorporate additional covariates (EMOS-GB, QRF), and neural network-based approaches (DRN, BQN, HEN). While the first two groups fit a separate model for each station, the neural networks estimate one locally adaptive model for all stations. Further, the models differ in the resulting forecast distributions, which range from parametric (EMOS, EMOS-GB, DRN) to semi- and non-parametric types (all others). The postprocessing models are trained on the period from 2010 to 2015 using the same configuration as in \citet{schulz2022machine} and evaluated on the remaining period from 2016 to 2023.
	
	To improve the forecast quality, NWP models are continuously developed further and frequently updated. Within the period from 2010 to 2023, the underlying NWP model has undergone several updates, of which three change the systematic errors of the ensemble forecasts drastically (on 22/03/2017, 16/05/2018 and 10/02/2021). These model updates present a challenge for statistical postprocessing methods, as the corrections learned on previous NWP model versions may not lead to the same improvements when applied to the current model version. In essence, the training data becomes less representative of the test data after each update. 
	The methods considered here have been trained on data until the end of 2015, which is before the first major model update in 2017. Still, they are applied to ensemble forecasts that have been generated by another NWP model version, which has undergone up to three major updates. Thus, we expect the behavior both of the ensemble predictions and the postprocessed forecasts to change systematically over time. For further details on the data and postprocessing methods, we refer to \citet{schulz2022machine} and \citet{primo2024comparison}.

	\subsubsection{Methods and implementation details}
	
	We assess the $m=9$ different forecasting models with the CRPS and construct SMCSs for the weak hypothesis as proposed in Section \ref{Subsec:SMCS_for_weakly_superior_objects}. Since the forecasts are predictable, the CRPS differences are conditionally bounded, see Supplement \ref{appendix:conditionally_bounded_scores}. We use the transformation we discussed in Section \ref{Subsec:Predictable_bounds_instead_of_fixed_bounds} to obtain uniformly bounded score differences. That is, $\lvert d_{ij,t}^k\rvert \leq c/2$ for the resulting CRPS differences at station $k\leq K$, for $c=2$, $i,j\leq m$, $t\in \N$. We set $\lambda_{ij,t}=(2c)^{-1}=1/4$ and $\alpha=0.1$. For inference on all stations simultaneously, we consider the average scaled CRPS differences ${d}_{ij,t}=\sum_{k=1}^{K}{d}^k_{ij,t}/K$, for $i,j\leq 9$ and $t\in \N$.

	\subsubsection{Results}

	First, for the SMCS with respect to the loss differences averaged over all stations, our conclusions regarding the performance of postprocessing methods align with the results in \citet{schulz2022machine}. The left panel of Figure \ref{figure:Wind_gust_merged_results} displays the SMCS and shows that the two neural network-based approaches, DRN and BQN, which performed best in the former study, are included in the SMCS for the entire time period. The other models fall out of the SMCS over time, where the order coincides with that found in the former study, e.g., the  EPS is omitted first.
	In this case, where we average the performances over all stations, the SMCS behaves as for the stronger hypotheses, that is, once we eliminate a model, we do not include it anymore.
	
	Second, we consider the SMCS station-wise. The right panel of Figure \ref{figure:Wind_gust_merged_results} gives a summary, as it shows the number of stations for which a forecasting model is included in the corresponding SMCS. The EPS is excluded from most SMCSs already within the first year, the basic techniques EMOS, MBM and IDR are also excluded relatively fast, and that the machine learning approaches EMOS-GB and QRF fall out of the SMCSs successively over time, until early 2021 when the last NWP model update occurred. Interestingly, the number of stations that include the HEN network approach increases after the last NWP model update, a behavior only detectable using the weak hypothesis. 
	Finally, the BQN and DRN approach are included at almost all stations over the entire period, only with DRN being excluded over the last year at some stations. 
	
	\begin{figure}[t]
		\centering
		\includegraphics[width= 0.49\textwidth]{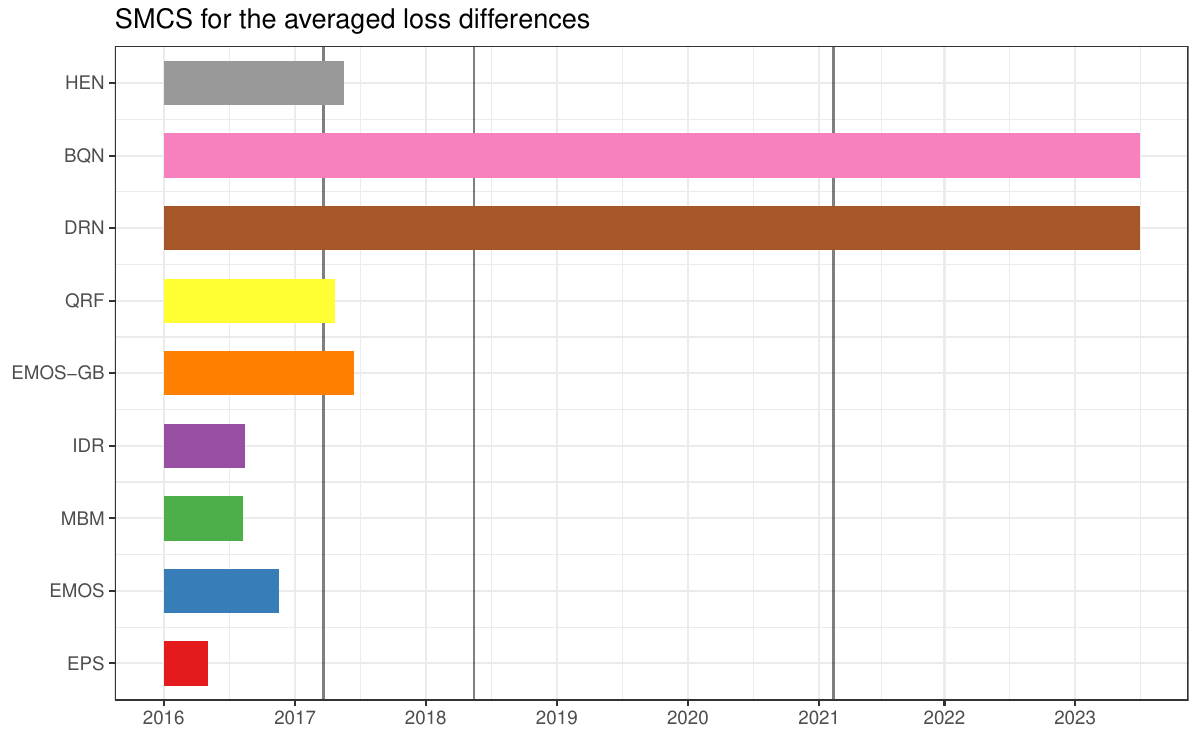}
		\includegraphics[width= 0.49\textwidth]{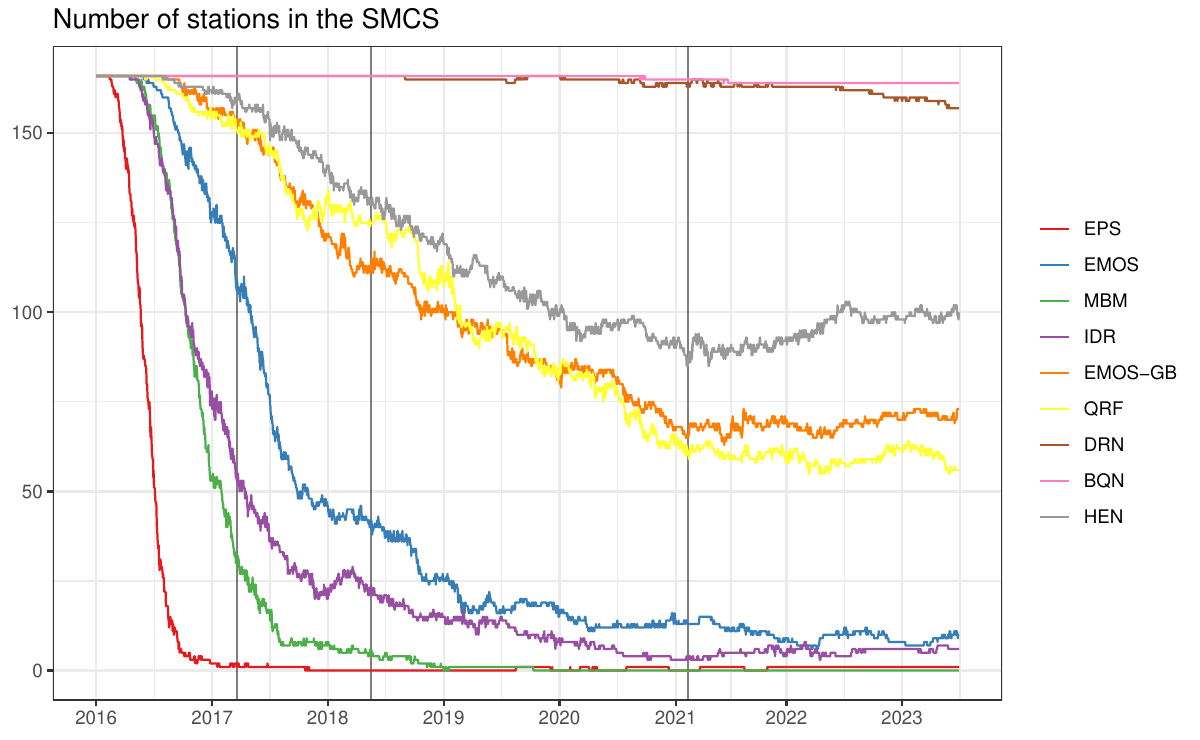}
		\caption{SMCS of averaged loss differences over all stations (left) and the number of stations where a method is included in the SMCS dependent on time (right). The vertical black lines indicate the three major NWP model updates.}
		\label{figure:Wind_gust_merged_results}
	\end{figure}
	
	Taking a closer look at the individual stations, we observe different kinds of local behavior that is not visible when averaging over the loss differences. We can broadly divide the stations into three categories: Regular (56 stations), dynamic (55) and fluctuating(55). Figure \ref{fig:Wind_gust_station_categories} shows one exemplary station for each of the three categories as well as the mean size of the station-wise SMCSs over time.
	The SMCS of a \emph{regular} station behaves (almost) equally as for the stronger hypotheses, which is also seen in the plot of the mean SMCS size over the time period. 
	The \emph{dynamic} stations show a different pattern. As for the regular stations, some of the inferior models are excluded after a certain time period. However, in contrast to the regular stations, some methods that have been eliminated before are now included again, coming closer to the end of the time period. This behavior becomes apparent when looking at the mean size of the SMCSs, which increases after the last update of the underlying NWP model, which systematically changes the predictive performance of the methods. These cases are interesting, since any test that would be performed just at the end of the evaluation period (or at some given time point in the middle) may not detect the deficiencies of these models. Consider, for example, the SMCS at station 10044, Leuchtturm Kiel, displayed at the top right of Figure \ref{fig:Wind_gust_station_categories}, where EMOS, IDR and HEN are excluded from the SMCS for long periods and become competitive again later on. Although these three models do not perform as strongly as DRN and BQN, the SMCS indicates that they are still better than the remaining models for this station. 
	Finally, we have the \emph{fluctuating} stations, where at least one model is repeatedly excluded and included from the SMCS. As for the other two groups, this behavior is reflected in the mean size of the SMCS, which is oscillating more for the fluctuating stations than the other two categories. In all cases, the sequential nature of the SMCS yields a deeper insight into the performance of the methods than tests for predictive ability which are issued at single time points only.

	\begin{figure}[t]
		\centering
		\includegraphics[width = 0.49\textwidth]{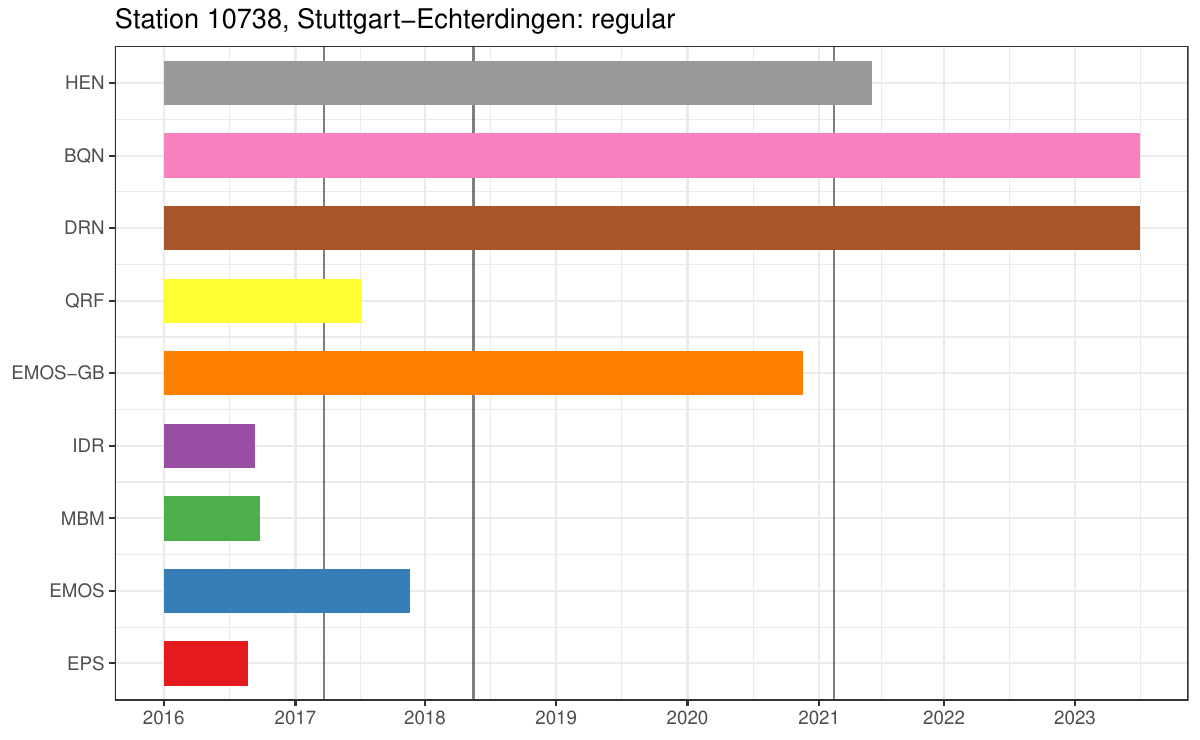}
		\includegraphics[width = 0.49\textwidth]{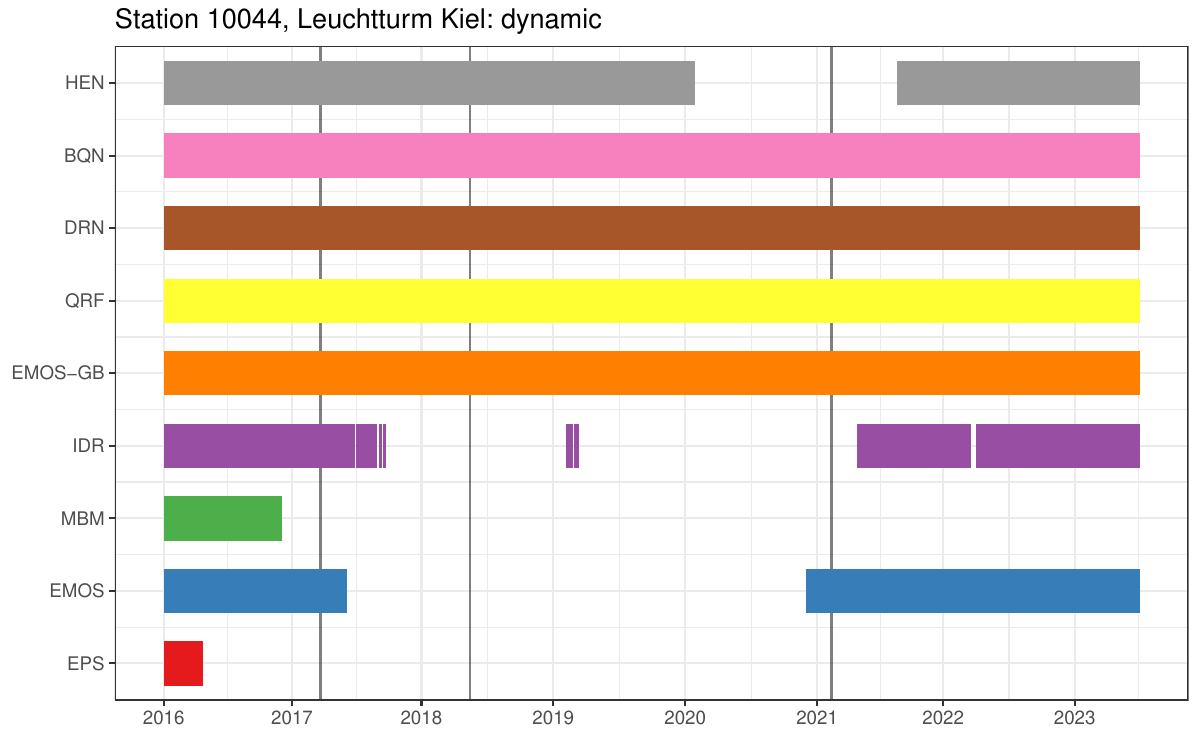}
		\includegraphics[width = 0.49\textwidth]{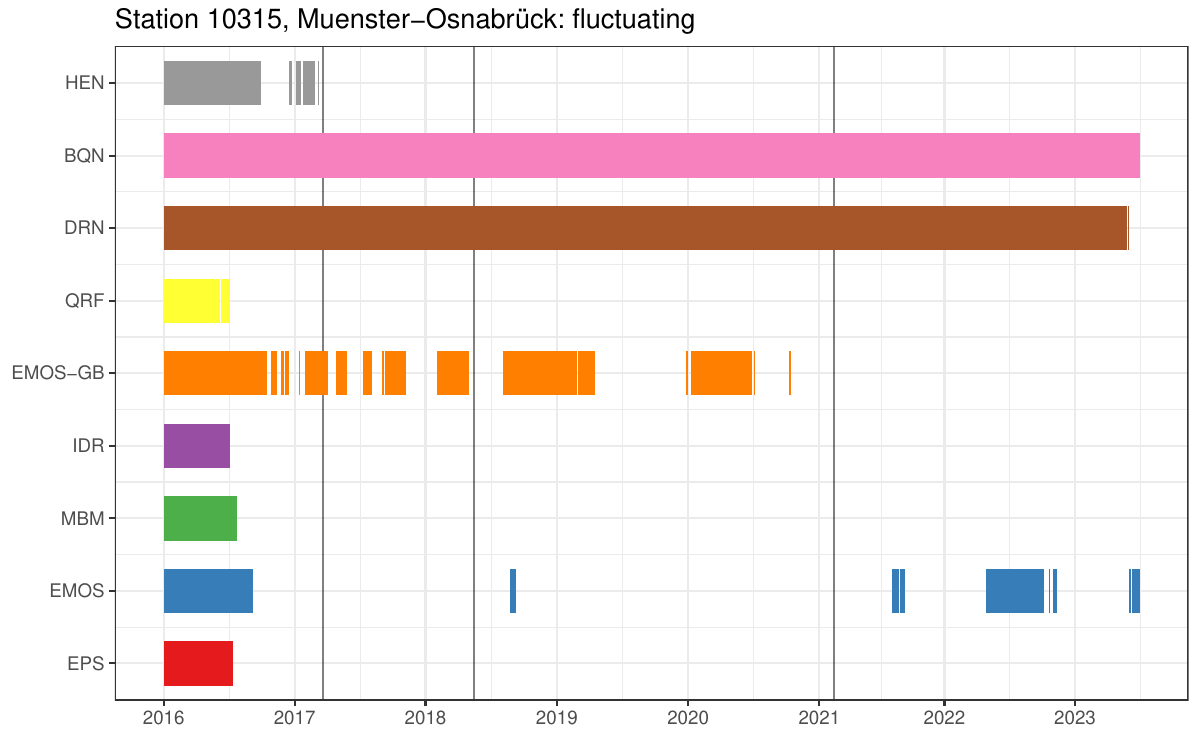}
		\includegraphics[width = 0.49\textwidth]{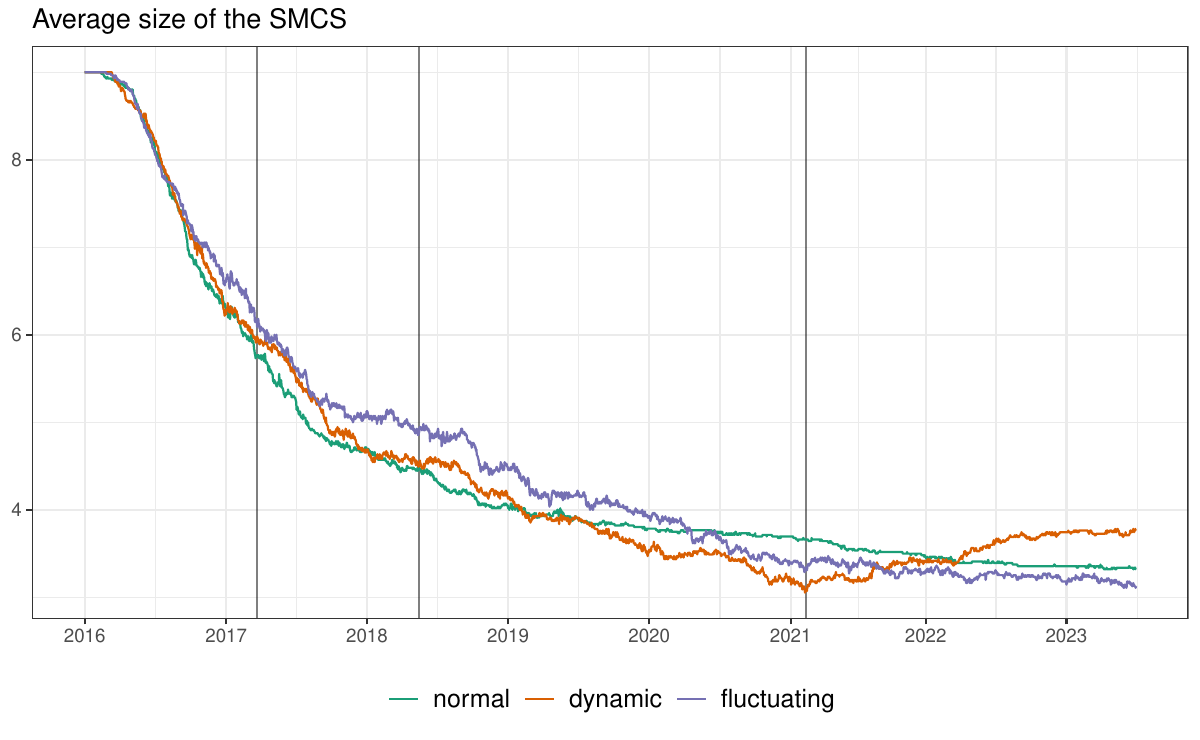}
		\caption{SMCSs for one exemplary station in each of the three categories (top left: normal; top right: dynamic; bottom left: fluctuating) and the mean size of the SMCS averaged over the stations of each category dependent on the time (bottom right). 
		} \label{fig:Wind_gust_station_categories}
	\end{figure}

	\section{Discussion}\label{Sec:Discussion}
	
	We introduce sequential model confidence sets and provide methodology to construct them with respect to three important notions of forecast superiority. SMCSs allow to continuously monitor the performance of some given statistical models. They incorporate the uncertainty in choosing models by being flexible in size, rely on minimal distributional assumptions, and achieve anytime-valid coverage guarantees. 
	
	We have provided coverage guarantees of all superior models with high probability uniformly over time. Other coverage guarantees could be of interest but require somewhat different approaches. The literature on inference on the argmin of a vector of means typically focuses on the guarantee that each superior model is covered with high probability. Future work could consider time-uniform extensions of such guarantees. Furthermore, one could aim for a time-uniform bound of the expected fraction of wrongly rejected models by the total number of rejected models. This can be achieved by time-uniform control of the false discovery rate and we give first results in this direction in Supplement \ref{appendix:FDR_control}. 
	
	We have mainly focused on (conditionally) bounded score differences. However, important loss functions such as the squared error or the logarithmic score do not satisfy this requirement. Nevertheless, our general theory is not limited to bounded score differences.
	We have provided a simulation study for SMCSs with the squared error loss in Supplement \ref{Appendix:MSE_Simulation}. Nonetheless, further research is warranted to extend the understanding and implementation of SMCSs to encompass more general loss functions. Finally, it would be interesting to examine how our results extend to important information criteria, and to apply sequential model confidence sets for sequential model selection.

	\section*{Data and replication material}
	Code in \textsf{R} for the case studies and replication material for the simulations are available at \url{https://github.com/GGavrilos/SMCS}.
	
	\section*{Acknowledgments}
	The comments of three anonymous referees significantly improved the quality of the paper. We thank Lutz D\"umbgen, Peter Gr\"unwald, and Tim Stephan for valuable discussions and inputs. Sebastian Arnold and Johanna Ziegel gratefully acknowledge financial support from the Swiss National Science Foundation. Computations for the simulations and parts of the wind gust case study have been performed on UBELIX (\url{https://ubelix.unibe.ch/}), the HPC cluster of the University of Bern.
	Benedikt Schulz gratefully acknowledges funding within the project C5 ``Dynamical feature-based ensemble postprocessing of wind gusts within European winter storms'' of the Transregional Collaborative Research Center SFB/TRR 165 ``Waves to Weather'' funded by the German Research Foundation (DFG).
	We thank Cristina Primo, Reinhold Hess and Sebastian Trepte for providing the forecast and observation data, and Robert Redl for assistance in data handling.

	\bibliographystyle{abbrvnat}
	\bibliography{biblio}
	
	\newpage
	
	\appendix
	\section*{Supplementary material}
	\section{Scoring rules and scoring functions}\label{appendix:forecast_evaluation}
	In this section, we review some important concepts for forecast evaluation. We assume that the unknown quantity of interest maps into some measurable space $\X$ and let $\PX$ be the family of all probability distributions on $\X$.  
	
	A \emph{statistical functional} is a map $T: \X \to 2^\X$. Point-valued forecasts for functionals should be compared using consistent scoring functions \citep{Gneiting2011a}. A measurable map $S : \X^2\rightarrow \bar{\R}$ is a \emph{consistent scoring function for the functional $T$ relative to the class $\mathcal{P} \subseteq \PX$} if for all $x \in \X$, $F \in \mathcal{P}$, the integral $\int S(x,y)\diff F(y)$ exists and, for all $t\in T(F)$, $\int S(t,y)\diff F(y) \leq \int S(x,y) \diff F(y)$. For real-valued outcomes, that is $\mathcal{X}=\R$, the \emph{Brier score} (or \emph{quadratic score}) $\BS(x,y)=(x-y)^2$, and the \emph{quantile score} $\QS_\tau(x,y) =(\one{\{y \leq x\}}-\tau)(x-y)$ are main examples of consistent scoring functions for the mean functional, and for the quantile at level $\tau \in (0,1)$, respectively. A scoring function $S$ is consistent for the quantile functional at level $\tau \in (0,1)$ if and only if it is of the form $S_\tau(x,y)= (\one{\{y\leq x\}}-\tau) \big(g(x)-g(y)\big), x,y \in \R, $ for some increasing function $g:\R\to \R$. Similar characterization result for the mean and expectiles can be found in \citet{Gneiting2011a} alongside historical references for them. 
	
	Probabilistic forecasts quantify the uncertainty of the future event and are specified as a probability measure over all possible outcomes. Probabilistic forecasts should be compared and evaluated using proper scoring rules \citep{Gneiting2007a,WaghmareZiegel2025}. A \emph{proper scoring rule} is a function $\myS: \mathcal{P}\times \X \mapsto \bar{\R}$ for some class $\mathcal{P}\subseteq \PX$ such that $\myS(F, \cdot)$ is measurable for any $F\in \mathcal{P}$, the integral $\int\myS(G,y) \diff F(y)$ exists, and for all $F,G \in \mathcal{P}$, $\int\myS(F,y) \diff F( y) \leq \int\myS(G,y) \diff F( y)$. 
	That is, for a proper scoring rule $\myS$, the expected score $\E_{Y \sim F} \myS(G,Y)$ is minimized with respect to all distributional forecasts $G\in \P$, if we correctly forecast the true distribution $F$ of the random variable $Y$. The most commonly used proper scoring rules for the evaluation of real-valued outcomes are the logarithmic score, $S(F,y) = -\log f(y)$, where $f$ is the density of $F$, and the \emph{continuous ranked probability 
		score} (CRPS; \cite{Matheson1976}), which is defined for all $F\in \PR$ with finite first moment and $y\in \R$ as
	\begin{equation*}
		\crps(F,y)=\int (F(x)-\one\{y\leq x\})^2 \diff x. 
	\end{equation*}
	The CRPS is a continuous extension of the Brier score for binary events and is popular across application areas and methodological communities; see, e.g., \citet{Gneiting2005b}, \cite{IDR}, and \citet{WaghmareZiegel2025} for a comprehensive recent review.
	
	Propriety of scoring rules is maintained under normalization by predictable quantities in the following sense: Let $(\F_t)_{t\in \N}$ be some filtration and $S$ be a proper scoring rule. Then, $\tilde{S}_t=S/c_t$ is conditionally proper for any predictable  $(c_t)_{t\in \N}\subseteq (0,\infty)$, that is, for any predictable sequence of probabilistic forecasts $(G_t)_{t\in \N}$, and any adapted $(Y_t)_{t\in \N}$ with conditional distribution $F_t$ given $\F_{t-1}$, we have 
	\begin{align*}
		\E[{\tilde{S}_t(F_t,Y_t)}\mid  \F_{t-1}]
		=c_t^{-1}\E\left[S(F_t,Y_t)\mid \F_{t-1}\right]\leq c_t^{-1}\E\left[S(G_t,Y_t)\mid \F_{t-1}\right]
		=\E[{\tilde{S}_t(G_t,Y_t)} \mid \F_{t-1}],
	\end{align*}
	where the inequality follows by propriety of $S$ applied conditionally on $\F_{t-1}$.

	\section{Conditionally bounded score differences}\label{appendix:conditionally_bounded_scores}
	The following results show that CRPS and quantile score differences are conditionally bounded due to the predictability of the forecasts. 
	
	\begin{lem}
		For $x_1, x_2 \in \R$ and $\tau \in (0,1)$, the difference $\delta(y) = \QS_\tau(x_1,y)-\QS_\tau(x_2,y)$	is bounded in $y \in \R$ with $(\tau-1)\lvert x_1-x_2\rvert\leq\delta (y) \leq \tau \lvert x_1-x_2\rvert.$
	\end{lem}
	
	\begin{proof} If $x_1\leq x_2$, then $\delta(y) 
		= y \one_{\{x_1 < y \leq x_2\}} - \tau (x_1-x_2) + x_1 \one_{\{y \leq x_1\}} - x_2\one_{\{y \leq x_2\}}$ and the bounds follow directly by differentiating cases. The case $x_1> x_2$ is analogous.
	\end{proof}
	
	\begin{lem}\label{lem:CRPS_cond_bounded}
		For $F_1,F_2\in \PR$, the CRPS difference $\delta(y) = \crps(F_1,y)-\crps(F_2,y)$
		is bounded in $y \in \R$, and attains its minimum and maximum at a crossing point of $F_1$ and $F_2$ (including $\{-\infty, + \infty\})$, respectively. 
	\end{lem}
	
	\begin{proof}
		We have $\delta(y)
		= \int F_1(x)^2-F_2(x)^2 \diff x + 2\int_{y}^\infty F_2(x)-F_1(x) \diff x,$
		where the first summand is independent of $y\in \R$. The function $x \mapsto F_2(x)-F_1(x)$ is bounded in $[-1,1]$, with limit $0$ as $x \to \pm \infty$ and changes sign at the crossing points of $F_1$ and $F_2$. 
	\end{proof}
	
	\begin{rem} The minimum and maximum in Lemma \ref{lem:CRPS_cond_bounded} is practically computable in application as illustrated below with the wind gust case study from Section \ref{Subsec:Windgust}. Here, the forecasts are either given as parametric (truncated logistic) distributions (EMOS, EMOS-GB, DRN), as distributions with finite support (ENS, MBM, IDR, HEN), or reported as quantiles on a fine grid of quantile levels (QRF, BQN), where the latter two groups may be considered as the same from a computational point of view. Thus, two CDFs may have infinitely many crossing points if they are either identical or if two (discrete) distribution are equal on some given interval. In the former case, the loss difference is trivially equal zero, and in the later case it is straightforward to check that, on the given interval, the absolute value of the CRPS difference is maximized at the boundaries of the interval. In either case, it is thus sufficient to consider a finite number of crossing points. Moreover, an analysis across all stations shows that the values of the maximum bounds are generally around 10, and only rarely reach about 20. Typically, we get these maximum bounds if we compare ENS with some other model, since the ENS forecasts' relatively large jumps result in big CRPS differences (area between the CDFs) when comparing them to (approximately) smooth distributional forecasts.
	\end{rem}
	
	\section{Multiple testing and proof of Theorem \ref{thm_SMCS_by_e-processes}}\label{app:pSMCS}
	
	We give some preliminaries on multiple testing corrections with e-processes before stating a proof for Theorem \ref{thm_SMCS_by_e-processes}. 
	
	Following \citet{Vovk_Wang_2021}, we call $f$ a \emph{symmetric e-merging function} if it is invariant with respect to permutations of its arguments and if $f(E_i, i \in I)$ is an e-value for any family of e-values $\{E_i \mid i \in I\}$ and any finite index set $I$. The most important e-merging function is the arithmetic mean as it essentially dominates any symmetric e-merging function \citep[Proposition 3.1.]{Vovk_Wang_2021}. The function $f$ is defined on $\bigcup_{i=1}^\infty \R^i$. In other words, it can be seen as a collection of functions on $\R^i$ for each $i \in \mathbb{N}$.
	
	In the spirit of \cite{Vovk_Wang_2021}, we say that non-negative stochastic processes $(E_{1,t}^\star)_{t\in \N}, \dots, (E_{m,t}^\star)_{t\in \N}$ for the hypotheses $\H_1, \dots, \H_m$ are \emph{family-wise valid (FWV)} if there exists a \emph{dominating family of e-processes}, that is, if there exists a family $\{(E^\Q_t)_{t\in \N}\mid \Q \in  \probM\}$ such that $(E^\Q_t)_{t\in \N}$ is an e-process for $\{\Q\}$ for all $\Q \in \probM$, and for any $i=1, \dots, m$ with $\Q\in \H_i$, it holds that \(E_t^\Q\geq E_{i,t}^\star\). Note that each process $(E^\star_{i,t})_{t\in\mathbb{N}}$ is an e-process for $\mathcal{H}_i$ by monotonicity.
	
	Let $(E_{1,t})_{t \in \N}, \dots, (E_{m,t})_{t \in \N}$ be e-processes for some hypotheses $\H_1, \dots, \H_m \subseteq \probM$, respectively, and let $f$ be a symmetric e-merging function. Then, by the closure principle \citep{Markus_etal_1976}, for any $t \in \N$, $E^{\star}_{i,t} = \min\{f(E_{j,t}, j \in I) \mid I \subseteq \{1, \dots, m\}, i \in I\}, i=1, \dots, m,$ are family-wise valid e-variables for the hypotheses $\H_1, \dots, \H_m$, see \citet[Section 5]{Vovk_Wang_2021}. That is, there exists a dominating family of e-variables $\{E_t^\Q \mid \Q\in \probM\}$, with $E_t^\Q \geq E^{\star}_{i,t}$ for any $i=1, \dots, m$ and $\Q \in \H_i$. Indeed, $E_t^\Q$ is explicitly given by $E_t^\Q = f( E_{i,t}, i \in I_\Q)$, where $I_\Q = \{i \mid \Q \in \H_i\}$. (If $I_\Q = \emptyset$, then $E_t^\Q = 1$.)
	
	Suppose that we adjust the e-processes $(E_{1,t})_{t \in \N}, \dots, (E_{m,t})_{t \in \N}$ by the closure principle with respect to the same e-merging function $f$ for all $t\in \N$. Then it follows that $(E^\star_{1,t})_{t \in \N}, \dots, (E^\star_{m,t})_{t \in \N}$ are family-wise valid e-processes for the hypotheses $\H_1, \dots, \H_m$, since $(E_t^\Q)_{t \in \N}$ are e-processes for \(\{\Q\}\) for all \(\Q\in \probM\). To see this, let $\tau$ be a stopping time. Then, $\E_\Q(E_\tau^\Q) = \E_\Q(f( E_{i,\tau}, i \in I_\Q)) \le 1$, since $E_{i,\tau}$, $i \in I_\Q$ are e-values under $\Q$.

	\begin{proof}[Proof of Theorem \ref{thm_SMCS_by_e-processes}]
		Let $\Q \in \probM$. We have  $i\in \M^{\bullet, \star}$ if and only if $\Q \in \H^{\bullet}_{i\cdot}$. Therefore,
		\begin{align*}
			\Q(\exists t \in \N: \M^{\bullet, \star} \not\subseteq \widehat{\M}_t ) 
			&=\Q(\exists t \in \N:E_{i\cdot,t}^{\star} \geq 1/\alpha \textrm{ for some }i \in \M^{\bullet, \star})\\
			&= \Q\left(\bigcup_{t \in \N} \bigcup_{i \in \M^{\bullet,\star}}  \{E_{i\cdot,t}^{\star} \geq 1/\alpha\}\right)\\
			&\leq \Q\left(\bigcup_{t \in \N} \bigcup_{i \in \M^{\bullet,\star}}  \{E_{t}^{\Q} \geq 1/\alpha\}\right) = \Q(\exists t \in \N:E_{t}^{\Q} \geq 1/\alpha) \le \alpha,
		\end{align*}
		where $E_t^\Q = f(E_{i\cdot,t}, i \in I_\Q)$ with $I_\Q = \{i \mid  \Q \in \H^\bullet_{i\cdot}\}$. It follows directly that the running intersection $ \widetilde{\M}_t= \bigcap_{r\leq t} \widehat{\M}_r, t\in \N$, is an SMCS as well.
	\end{proof}

	\begin{rem}
		For the construction of SMCSs in Section \ref{Subsec:SMCS_by_p_processes}, we merged the pairwise e-processes \((E_{ij,t})_{t\in \N}\) into a single e-process \(E_{i\cdot,t}^\star\) by averaging over \(j\), and then adjusted the merged e-processes using the closure principle.
		Alternatively, we could directly work with the pairwise e-processes without first merging them.
		More specifically, starting with the pairwise e-processes \((E_{ij,t})_{t\in \N}\), for all pairs \((i,j)\in \M_0\), one can directly use the closure principle \citep[Algorithm~1]{Vovk_Wang_2021} to derive the adjusted e-processes
		
		\begin{equation*}
			E^\star_{ij,t}=\min_{I \subseteq \{1, \dots, m\}^2: (i,j)\in I} \frac{1}{\lvert I \rvert} \sum_{(k,\ell) \in I} E_{k\ell,t}, \quad i,j\in \M_0, t \in \N.
		\end{equation*}
		Finally, for some significance level $\alpha\in (0,1)$, we would define
		\begin{equation*}
			\widehat{\M}_t= \left\{i \in \M_0 \mid  E_{ij,t}^{\star} <  1/\alpha, \text{ for all } j\in \M_0 \right \}, \quad t\in \N.
		\end{equation*}
		However, the corresponding tests would be less powerful, since for any $i\in \M_0$ and $j\neq i$,
		\begin{align*}
			\allowdisplaybreaks
			E^{\star }_{ij,t} &= \min\left\{\frac{1}{\lvert I \rvert }\sum_{(k,l)\in I}E_{kl,t}\;\Big\vert\; I \subseteq \{1, \dots, m\}^2, (i,j)\in I \right\} \\
			& \leq \min\left\{\frac{1}{\lvert I \rvert (m-1) }\sum_{k\in I} \sum_{l =1}^m E_{kl,t} \;\Big\vert\; I \subseteq \{1, \dots, m\}, i \in I \right\} = E_{i \cdot}^{\star}. 
		\end{align*}
	\end{rem}
	
	\section{SMCSs for marginal coverage}\label{appendix:marginal_coverage}
	In this section, we illustrate how our methods readily extend to the weaker marginal coverage guarantee where the multiple testing correction from the previous section may be omitted. 
	
	In analogy to Section \ref{Sec:SMCS}, we call we call a sequence $(\check{\M}_t)_{t \in \N} \subseteq \M_0$ a \emph{marginal coverage SMCS} for the superior objects $\M^{\bullet, \star}$, $\bullet \in \{\mathrm{s},\mathrm{uw}\}$ at level $\alpha\in (0,1)$ if, for any $\Q\in \probM$, 
	\begin{equation}\label{eq:def_marginal_coverage}
		\Q(\forall t\geq 1: i \in \check{\M}_t  ) \geq 1-\alpha \quad \textrm{for all }i\in \M^{\bullet, \star}.
	\end{equation}
	We refer to condition \eqref{eq:def_marginal_coverage} as time-uniform \emph{marginal} coverage. Clearly, this is a weaker requirement than given at \eqref{eq:def_coverage}, forming a time-uniform \emph{simultaneous} coverage.
	
	As in Section \ref{Subsec:SMCS_by_p_processes}, for $\bullet \in \{\mathrm{s},\mathrm{uw}\}$, assume that, for any $i,j\in\M_0$, $(E_{ij,t})_{t \in \N}$ is an e-process for the hypothesis $\H^\bullet_{ij}$, and let $E_{i\cdot,t} =  1/(m-1)\sum_{j \neq i} E_{ij,t}$ be the arithmetic mean e-process for the intersection hypothesis $\H^{\bullet}_{i\cdot} = \cap_{j\neq i} \H^\bullet_{ij}$.
	For $\alpha\in (0,1)$, define
	\begin{equation}\label{eq:def_SMCS_strong_hypthesis_e_process_marginal_coverage}
		\check{\M}_t= \left\{i \in \M_0 \mid  E_{i\cdot,t} <  1/\alpha 
		\right \}, \quad t\in \N.
	\end{equation}
	Then, the following adaption of Theorem \ref{thm_SMCS_by_e-processes} follows directly.
	
	\begin{thm}\label{thm_marginal_coverage_SMCS_by_e-processes}
		For any $\alpha \in (0,1)$, the sequence $(\check{\M}_t)_{t \in \N}$ defined at \eqref{eq:def_SMCS_strong_hypthesis_e_process_marginal_coverage} is a marginal coverage SMCS at level $\alpha$ for $\M^{\bullet, \star}$, $\bullet \in \{\mathrm{s},\mathrm{uw}\}$, and so is its running intersection. 
	\end{thm}
	
	\begin{rem}
		If $|\M^{\bullet, \star}|=1$, $\bullet \in \{\mathrm{s},\mathrm{uw}\}$, then \eqref{eq:def_coverage} and \eqref{eq:def_marginal_coverage} coincide, and we do not have to correct for multiple testing if there exists exactly one superior object. 
	\end{rem}
	
	Since weakly superior objects can change over time we have to adapt the definition of marginal coverage at \eqref{eq:def_marginal_coverage} for weakly superior objects. For each $i\in \M_0$, we let \begin{equation*}
		I^\star(i) = \{t \in \N \mid \Delta_{ij,t}\leq 0 \textrm{ a.s. for all }j\in \M_0\}, 
	\end{equation*}
	and call a sequence $(\tilde{\M}_t)_{t \in \N} \subseteq \M_0$ a 
	\emph{marginal coverage SMCS} for the weakly superior objects $(\wsupM_t)_{t\in \N}$,  at level $\alpha\in (0,1)$ if, for any $\Q\in \probM$, 
	\begin{equation*}
		\Q(\forall t\in I^\star(i) : i \in \tilde{\M}_t  ) \geq 1-\alpha \quad \textrm{for all }i\in \M_0.
	\end{equation*}
	
	Next, we show how one can construct marginal coverage SMCSs for the weakly superior objects. For $i \in \M_0$, define $\bDelta_{i,t} = (\Delta_{ij,t})_{j=1}^m$ and, for $\bX =(x_i)_{i=1}^m\in \R^m$,
	\[
	M_t^{(i)}(\bX) = \frac{1}{m} \sum_{j=1}^m M_{ij,t}(x_j),
	\]
	where $M_{ij,t}(x)$ is such that $(M_{ij,t}(\Delta_{ij,t}))_{t \in \N}$ is a test supermartingale for any $\Q \in \mathfrak{B}(\Omega)$. Then, $(M_t^{(i)}(\bDelta_{i,t}))_{t \in \N}$ is a test supermartingale for any $\Q \in \mathfrak{B}(\Omega)$. 
	For $\alpha \in (0,1)$, we define $m$ confidence regions
	\[
	C_t^{(i)} = \left\{\bX \in \R^m \mid M_t^{(i)}(\bX) < \frac{1}{\alpha}\right\}, \quad i \in \M_0,
	\]
	and set
	\begin{equation}\label{eq:marginal_coverage_SMCS_weakly_superior_objects}
		\tilde{\M}_t = \{i \in \M_0 \mid \forall j \not=i: C_t^{(i)}\cap \R_{j,-}^m \not= \emptyset\},
	\end{equation}
	where $\R_{j,-}^m = \{\bX \in \R^m \mid x_j \leq 0\}$. This means that, analogously to Section \ref{Subsec:SMCS_for_weakly_superior_objects}, we exclude $i$ from $\tilde{\M}_t$ if there exists a $j \not=i$ such that we may reject $\Delta_{ij,t} \le 0$, which is the case, if and only if $C_t^{(i)}\cap \R_{j,-}^m = \emptyset$ for some $j \not= i$. 
	\begin{thm} For any $\alpha \in (0,1)$, the sequence $(\tilde{\M}_t)_{t \in \N}$ defined at \eqref{eq:marginal_coverage_SMCS_weakly_superior_objects} is a marginal coverage SMCS for the weakly superior objects $(\wsupM_t)_{t\in \N}$. 
	\end{thm}
	
	\begin{proof}
		For any $\Q \in \mathfrak{B}(\Omega)$, and any $i \in \M_0$, we have
		\begin{align*}
			\Q(\exists t \in I^\star(i): i \notin \tilde{\M}_t) &=  \Q(\exists t \in I^\star(i), \exists j \in \M_0 : C_t^{(i)}\cap \R_{j,-}^m = \emptyset)\\
			&\leq \Q(\exists t \in I^\star(i) : M_t^{(i)}(\bDelta_{i,t}) \geq 1/\alpha) \leq \alpha\,.
		\end{align*}
	\end{proof}

	In Simulations 1 and 2, there is always an unique superior object, so marginal and uniform coverage are equivalent.
	Therefore, we can compute the SMCS without having to adjust the e-processes.
	Figure \ref{fig: no_adjustment} shows that this results in smaller confidence sets. Analogously, we compare the marginal coverage SMCS for the weakly superior objects given by \eqref{eq:marginal_coverage_SMCS_weakly_superior_objects} with the original SMCS from Section \ref{Subsec:SMCS_for_weakly_superior_objects}. This comparison is illustrated in Figure \ref{fig:Simulation3_marginal_coverage_results}.
	As expected, the average size of the SMCS given by \eqref{eq:marginal_coverage_SMCS_weakly_superior_objects} is generally smaller.
	
	\begin{figure}[ht]
		\centering
		\includegraphics[width=0.75\linewidth]{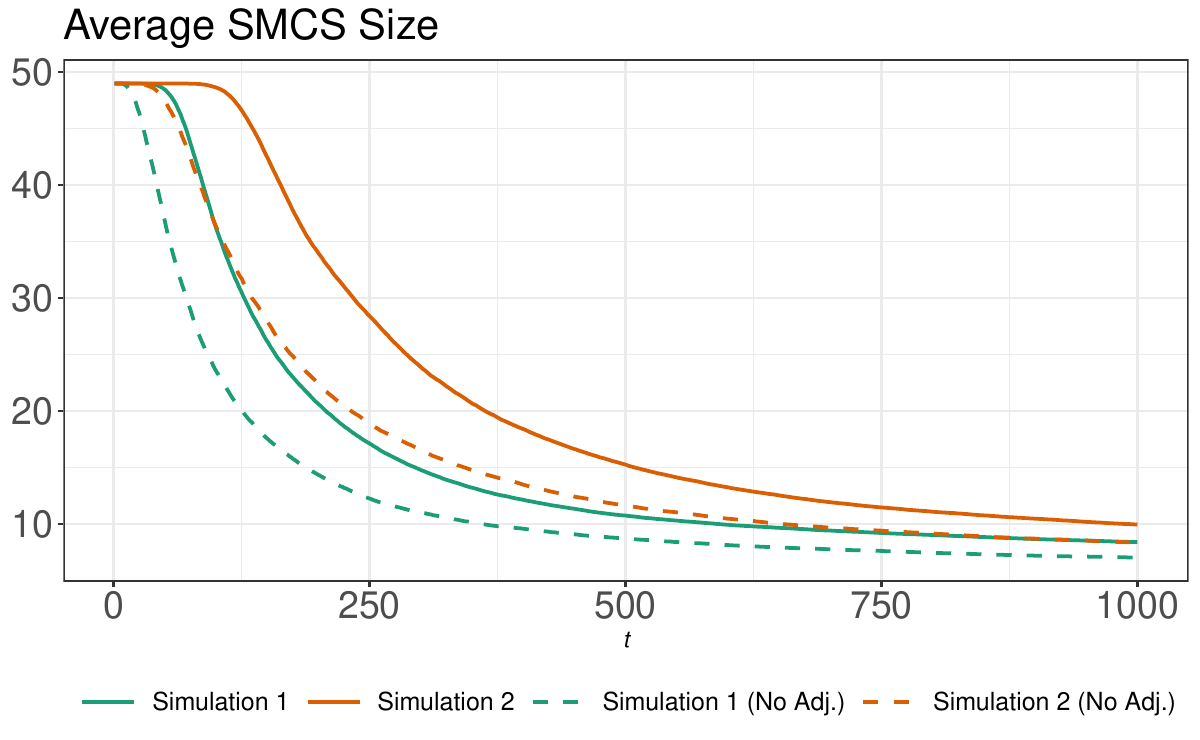}
		\caption{The SMCS shrinks slightly when uniform coverage is replaced by marginal coverage.
			In Simulations 1 and 2, marginal coverage and uniform coverage are equivalent, due to the existence of a unique superior object.}
		\label{fig: no_adjustment}
	\end{figure}

	\begin{figure}[h]
		\centering
		\includegraphics[width =0.8 \textwidth]{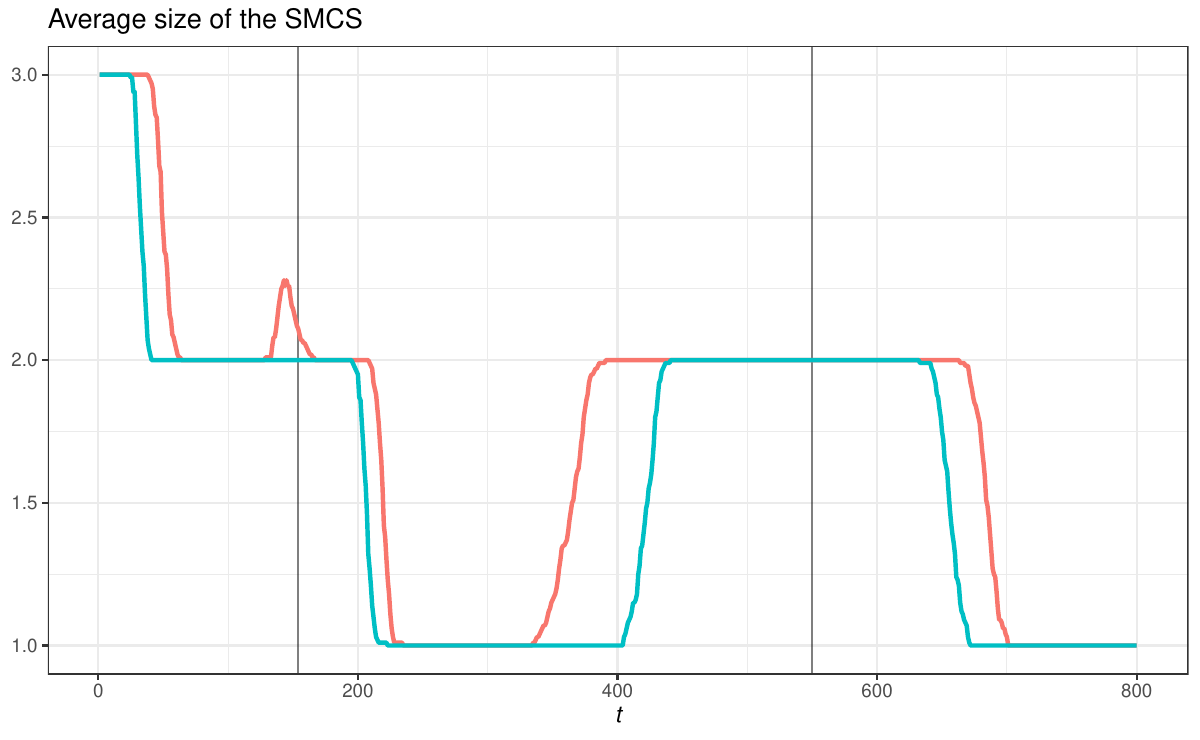}
		\caption{Size of the SMCS over time in Simulation 3, averaged over $100$ simulations, for marginal coverage (blue) compared to the original approach from Section \ref{Subsec:SMCS_for_weakly_superior_objects} (red), see the original Figure \ref{fig:Simulation3_results}. Also for the marginal coverage SMCS we have 100\% coverage of the superior object at all time points.}
		\label{fig:Simulation3_marginal_coverage_results}
	\end{figure}
	
	\section{SMCSs by FDR control}\label{appendix:FDR_control}
	
	For $i=1, \dots, m$, let $(\Psi_{i,t})_{t\in \N}$ be a sequential test for some hypothesis $\H_i\subseteq \probM$. We let $R_{t} = \sum_{i=1}^m \Psi_{i,t}$ be the number of all rejected hypotheses (discoveries) at $t\in \N$. For $\mathbb{Q} \in \bigcup_{i=1}^m \H_i$, we let $F_t(\mathbb{Q}) = \sum_{i \in I_\mathbb{Q}}\Psi_{i,t}$ be the number of all true null hypotheses that are rejected at time $t$. The \emph{false discovery rate (FDR)} at $t\in \N$ of the given testing procedures is defined as
	\begin{equation*}
		\textrm{FDR}_t=  \sup_{\Q \in \bigcup_{i=1}^m \H_i}  \mathbb{E}_{\mathbb{Q}} \left[\frac{F_t(\Q)}{\max\{1,R_t\}}\right]. 
	\end{equation*}
	We say that the sequential tests $(\Psi_{1,t})_{t\in \N}, \dots, (\Psi_{m,t})_{t\in \N}$ \emph{time-uniformly control the false discovery rate at level} $\alpha \in (0,1)$ if, for all (possibly infinite) stopping times $\tau$, $\textrm{FDR}_\tau \le \alpha .$ 
	\cite{Wang_Ramdas_FDR_control} study FDR control with e-values in a nonsequential setting. Looking carefully at the proof of their Theorem 2, we see that their arguments directly extend to e-processes to obtain sequential tests that time-uniformly control FDR as explained in the following. 
	
	Let $(E_{1,t})_{t \in \N}, \dots, (E_{m,t})_{t \in \N}$ be e-processes for the hypotheses $\H_1, \dots, \H_m \subseteq \probM$. Following \cite{Wang_Ramdas_FDR_control}, we let $E_{[1],t} \ge\dots \ge E_{[m],t}$ be the order statistics of $E_{1,t},\dots,E_{m,t}$ in decreasing order and define $i^\star_t = \max\{i \in \{1,\dots,m\} \;\vert\; {iE_{[i],t}} \ge m/{\alpha}\}$, for a given $\alpha \in (0,1)$.
	Then, we have a time-uniform FDR control at level $\alpha$ for the sequential tests $\Psi_{i,t}= (\one\{E_{i,t} \ge E_{[i_t^\star],t}\}), i=1, \dots, m, t \in \N.$ As our construction for the strong and uniformly weak hypothesis builds on e-processes, we can use the above sequential tests with respect to the merged e-processes $( E_{1\cdot, t})_{t\in \N}, \dots, ( E_{m\cdot,t})_{t\in \N}$ given in Section \ref{Subsec:SMCS_by_p_processes}.
	That is, at each $t\in \N$, we reject the $i^\star_t$ hypotheses with the largest values of the corresponding e-processes, and obtain the alternative SMCS $\check{\M}_t = \{i \in \M_0 \mid E_{i\cdot, t}< E_{[i_t^\star]\cdot, t} \}$.
	Importantly, this sequence does not necessarily satisfy the time-uniform coverage property from Section \ref{Sec:SMCS} anymore. In contrast, it bounds the expected ratio of wrongly rejected models to the total number of rejected models at any random time by $\alpha$, that is, for any $\Q\in \probM$ and any random stopping time $\tau$, we have $\mathbb{E}_{\mathbb{Q}} \big[ {\lvert \check{\M}^c_\tau \cap \supM_\tau \rvert/\max\{1,\lvert \check{\M}_\tau^c \rvert \}}\big] \le \alpha,$ where $\check{\M}^c_\tau= \{1, \dots, m\} \setminus \check{\M}_\tau$ and $(\supM)_{t \in \N}$ is given by \eqref{eq:def_strongly_superior_objects} or \eqref{eq:def_uniformly_weakly_superior_objects}.
	
	\section{Proofs for Section \ref{Subsec:SMCS_for_weakly_superior_objects}}\label{Appendix:proofs}
	\begin{proof}[Proof of Theorem \ref{thm:3.5}]
		For any $\Q\in \probM$, we have
		$\Q(\exists t \in \N: \wsupM_t \not \subseteq \widehat{\M}_t )\leq\Q(\exists t \in \N : \bDelta_{t} \notin C_t) \le \alpha$ by Ville's inequality. Moreover, we have
		\begin{align*}
			\Q(\forall t \geq 1: \uwsupM \subseteq \widetilde{\M}_t)&=
			\Q\left(\forall t \geq 1: \uwsupM\subseteq\widehat{\M}_t\right)\\
			&\geq 
			\Q\left(\forall t \geq 1: \wsupM_t\subseteq\widehat{\M}_t\right)\geq 1-\alpha.
		\end{align*}
	\end{proof}

	\begin{proof}[Proof of Proposition \ref{prop:convex_upper_set}]
		For $t\in \N$, we have $C_{t,1-\alpha}=\{\bX \in \diagM \mid \lVert M_{t}(\bX) \rVert_1 \leq m(m-1)/\alpha\}$, for $M_t: \diagM \to \diagM, \bX \mapsto \left(M_{ij,t}(X_{ij})\right)_{i\neq j}$
		where $\lVert \bX \rVert_1 = \sum_{i,j=1}^m \lvert X_{ij}\rvert $ for $\bX \in \R^{m\times m}$. Recall that, for any $i\neq j$, $M_{ij,t}$ is a nonnegative, convex and decreasing function. Thus, the second claim follows immediately since $\bA \leq \bB$ implies $M_t(\bA)\geq M_t(\bB)$. Let $\bA, \bB \in C_{t,1-\alpha}$, then, for any $\lambda \in [0,1]$,
		\begin{align*}
			\lVert M_t\big(\lambda \bA + (1-\lambda)\bB \big)\rVert_1
			&\leq \lVert \lambda M_t(\bA) + (1-\lambda) M_t(\bB) \rVert_1 \\
			&\leq \lambda \lVert  M_t(\bA)\rVert_1 + (1-\lambda) \lVert M_t(\bB) \rVert_1 \leq \frac{m(m-1)}{\alpha},
		\end{align*}
		where the first inequality holds by convexity and the nonnegativity of $M_t$, the second one by the axioms of a norm, and the third one by the fact that $\bA, \bB \in C_{t,1-\alpha}$. 
	\end{proof}
	
	\section{Further simulation results}
	
	\subsection{Uniformly weakly superior objects for unbounded losses}\label{Appendix:MSE_Simulation}
	
	In Supplement \ref{appendix:conditionally_bounded_scores}, it is shown that the score differences for quantile forecasts and probabilistic forecasts evaluated with the CRPS are conditionally bounded. In contrast, the highly relevant mean squared error (MSE) for mean forecasts does not satisfy this property.  In the following, we present a simulation example where the MSE differences are conditionally sub-exponential. By using a small number of models, we enable graphical monitoring of the e-processes as measure of evidence over time.
	
	For some data $(Y_t)_{t\in \N}$, we consider $m=9$ mean forecasters issuing mean-forecasts $$m_{i,t}= Y_t + X_{i,t}, \quad t\in \N,$$ with i.i.d.\ normally distributed errors $X_{i,t} \sim \mathcal{N}(\varepsilon_{i,t}, 1+ \delta_{i,t})$ independent of the the sequence $(Y_t)_{t\in \N}$. 
	Analogously to Simulation 2, we let $(\varepsilon_{i,t},\delta_{i,t}) \in \{-0.5, 0, 0.5\}^2$ be fixed over time, with the only exception that the superior model $i_0$ with zero bias and lowest variance is worse on Sundays, more precisely, $\delta_{i_0,t}=-0.5$ and $\varepsilon_{i_0,t}=0.6 \cdot \one_{\{t \in 7\N\}}$. If we assess the mean predictions by the MSE, we obtain the following conditional expected score differences, independent of the data $(Y_t)_{t\in \N}$,
	\begin{equation*}
		L_{ij,t}= X_{i,t}^2 - X_{j,t}^2, \quad \mu_{ij,t}= \E(d_{ij,t})=\varepsilon_{i,t}^2 + \delta_{i,t} - (\varepsilon_{j,t}^2 + \delta_{j,t}), \quad t\in \N. 
	\end{equation*}
	Recall that the square of a Gaussian random variable is sub-exponential, and the difference of sub-exponential random variables is again sub-exponential. Therefore, for all $i,j \in \M_0$, the loss differences $(d_{ij,t})_{t\in\N}$ are sub-exponential with constants $\nu_{ij}=2(1+\delta_i)+2(1+\delta_j)$ and $\alpha_{ij}=4(1+\max\{\delta_i,\delta_j\})$, that is, for all $t\in \N$,
	\begin{equation*}
		\E\left(e^{\lambda(d_{ij,t}-\mu_{ij,t})-\nu_{ij}^2 \lambda^2/2} \big |\F_{t-1}\right)\leq 1, \quad |\lambda|<1/\alpha_{ij}.
	\end{equation*}
	It follows that the cumulative product of the increments on the left-hand side is a non-negative supermartingale, and thus  
	\begin{equation}\label{e-processes_sub_exp_V2}
		E_{ij,t}= \exp\left\{\lambda_{ij} t\hat{\Delta}_{ij,t}-\frac{t}{2}\nu_{ij}^2\lambda_{ij}^2 \right\}, \quad t\in \N,
	\end{equation}
	is an e-process for $\H_{ij}^\textrm{uw}$, for any $0\leq \lambda_{ij}<1/\alpha_{ij}$. Note that the e-process at \eqref{e-processes_sub_exp_V2} could equivalently be expressed as in Proposition \ref{prop:e-process_choe_ramdas} with the variance process $V_{ij,t}=t\nu_{ij}$, see \citet[Proposition 5, Appendix E]{Howard_Chernoff_bounds} for a discussion of the two equivalent notions of the sub-exponentiality property. 
	We construct an SMCS for the uniformly weakly superior objects as given in Section \ref{Subsec:SMCS_by_p_processes} for $\alpha=0.1$. Figure \ref{fig:MSE_simulation_results} shows the average size of the SMCS over time with the averaged e-processes given in Figure \ref{fig:MSE_simulation_e-processes}. 
	
	\begin{figure}[h]
		\centering
		\includegraphics[width =0.5 \textwidth]{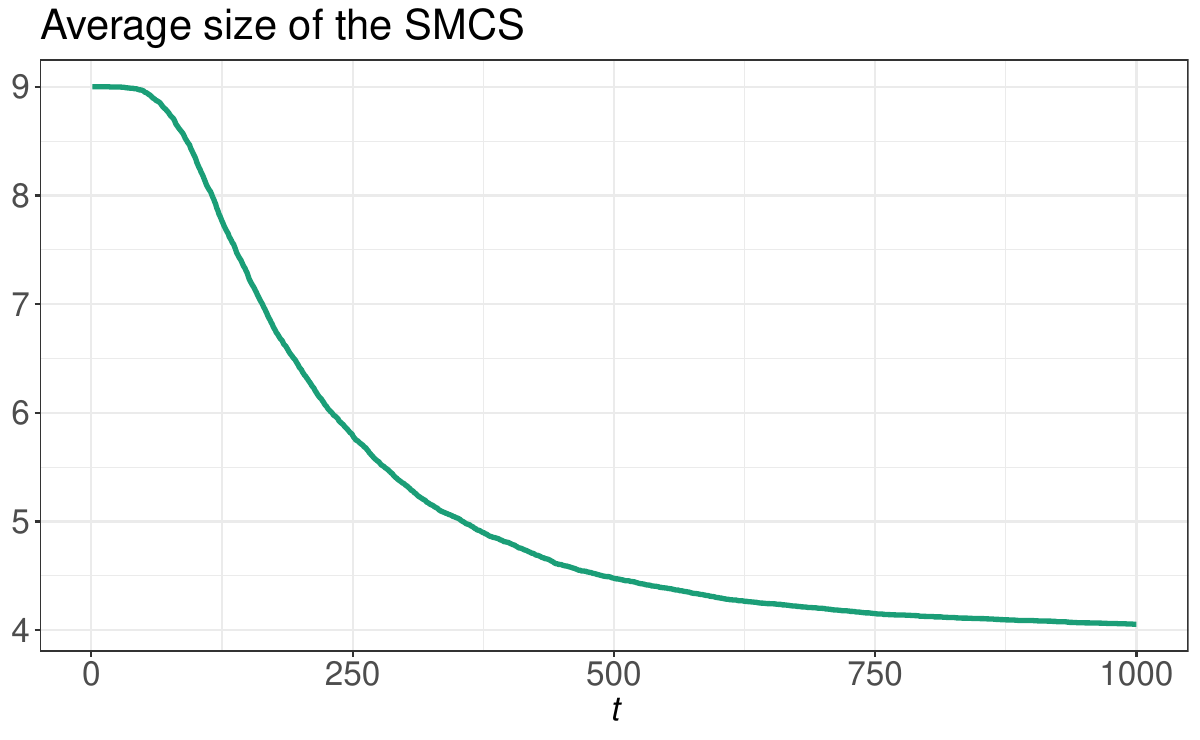}
		\caption{Size of the SMCS over time, averaged over $N=1000$ simulations. The SMCS never wrongly rejects the superior model $i_0$ and contains $4.06$ models on average at the end of the evaluation period.}
		\label{fig:MSE_simulation_results}
	\end{figure}

	\begin{figure}[h]
		\centering
		\begin{minipage}[t]{0.5\textwidth}
			\vspace{0pt} 
			\centering
			\includegraphics[width=\linewidth]{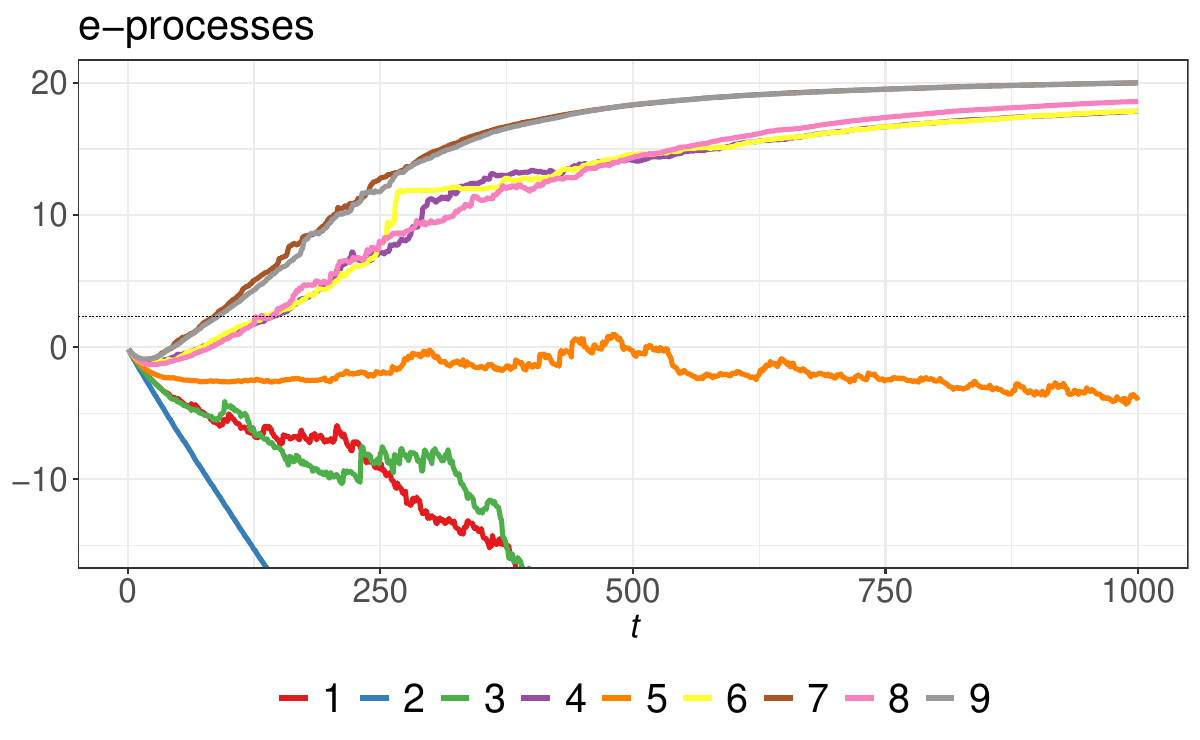}
		\end{minipage}
		\hfill
		\begin{minipage}[t]{0.45\textwidth}
			\vspace{0pt} 
			\centering
			\footnotesize
			\begin{tabular}{rrr}
				\textbf{Model} & $\varepsilon$ & $\delta$ \\
				\hline
				1 & $-0.5$ & $-0.5$ \\
				2 & $0.0/0.6$  & $-0.5$ \\
				3 & $0.5$  & $-0.5$ \\
				4 & $-0.5$ & $0.0$  \\
				5 & $0.0$  & $0.0$  \\
				6 & $0.5$  & $0.0$  \\
				7 & $-0.5$ & $0.5$  \\
				8 & $0.0$  & $0.5$  \\
				9 & $0.5$  & $0.5$  \\
				\hline
			\end{tabular}
		\end{minipage}
		\caption{Left: The averaged adjusted e-processes $(E_{i\cdot, t}^\star)_{t\in \N}$ on a logarithmic scale for the models $i=1, \dots, 9$ with the critical value $\log(1/\alpha)\approx 2.3$ given as a horizontal line. Right: Parameters of the models with superior model given by $i_0=2$.}
		\label{fig:MSE_simulation_e-processes}
	\end{figure}
	
	\subsection{Strongly superior objects for non-Gaussian data}\label{app:F.2}
	
	Simulations 1 and 2 in the main text showcase the ability of sequential model confidence sets to deal with sequences of dependent data, while preserving anytime-validity.
	In these simulations, our sequential methods also exhibit remarkable statistical power.
	As shown in Figure \ref{fig:Simulations1&2_results}, the original model set, which contains 49 different models, is narrowed down to less than 10 models after 1000 time steps.
	The data points are sampled from Gaussian distributions.
	One possible concern is that this design is not challenging enough, so the estimated statistical power might be overly optimistic.
	
	In this section, we provide some empirical results on the performance of our methods in a non-Gaussian setting.
	We also compare them with the original model confidence set (MCS) of \cite{MCS}.
	Even though the original MCS is not anytime-valid, it is still be interesting to compare its size to our anytime-valid SMCSs.
	The simulation design is similar to the one of Simulation 1: we set \(Y_0=0\) for convenience, and, for \(n=1000\), we sample \((Y_t)_{t=1}^n\) from the mixture of the Gaussian distributions
	\begin{equation*}
		\mathcal{N}\left(\arctan(Y_{t-1}), 1+\sqrt{|Y_{t-1}|}\right),\quad  \mathcal{N}\left(-\arctan(Y_{t-1}), 1+\sqrt{|Y_{t-1}|}\right)
	\end{equation*}
	with equal weights.
	The use of \(\arctan(\cdot)\) prevents against violent oscillations over different orders of magnitude, which complicates the computation of the predictable bounds \(c_{ij,t}\), and is rare in practice.
	The forecasting models are defined by introducing bias and dispersion parameters \(\varepsilon,\delta\), in the same spirit with Simulations 1 and 2: at each step \(t\geq 1\), forecaster \(i\) outputs the mixture \(F_{i,t}\) of the Gaussian distributions
	\begin{equation*}
		\mathcal{N}\left(\arctan(Y_{t-1}) + \varepsilon_i, 1+\sqrt{|Y_{t-1}|}+\delta_i\right),\quad  \mathcal{N}\left(-\arctan(Y_{t-1}) + \varepsilon_i, 1+\sqrt{|Y_{t-1}|} + \delta_i\right),
	\end{equation*}
	where \((\varepsilon,\delta)\) ranges over the set \(\{-0.6,-0.4,\ldots,0.4,0.6\}^2\).
	The CRPS loss at a point \(y\), for a Gaussian mixture \(F\) with equally weighted components \(\mathcal{N}(\mu_1,\sigma_1^2)\) and \(\mathcal{N}(\mu_2,\sigma_2^2)\) is given by the following formula \citep{Jordan_ScoringRules}:
	\begin{equation*}
		\frac{1}{2}A(y-\mu_1,\sigma_1^2)+\frac{1}{2}A(y-\mu_2,\sigma_2^2)-\frac{1}{2}A(0,2\sigma_1^2)-\frac{1}{2}A(0,2\sigma_2^2)-\frac{1}{2}A(\mu_1-\mu_2,\sigma_1^2+\sigma_2^2),
	\end{equation*}
	where \(A(\mu,\sigma^2)=\mu\left(2\Phi(\mu/\sigma)-1\right)+2\sigma\cdot \varphi(\mu/\sigma)\), and \(\Phi,\varphi\) denote the CDF and the density of the standard normal distribution, respectively.

	In contrast to Simulation 1, \(\mu_{ij,t}=\E\left[d_{ij,t}\mid \F_{t-1}\right]\) now depends on \(t\), so the relative performance of the different models can vary with time.
	This property makes this setting more realistic, but it raises an important question about the applicability of the MCS of \cite{MCS}.
	In particular, the theory in that paper is developed around the assumption that the unconditional mean \(\E[d_{ij,t}]\) does not depend on \(t\).
	A simulation of \(10000\) trajectories of the sequence \((Y_t)_{t=1}^{1000}\) and the corresponding forecasts \(F_{i,t}\) confirms that \(\E[L_{i,t}]\) does not have any clear temporal trends, see Figure \ref{fig:expected_CRPS}.
	Therefore, the assumption of a time-independent expected loss is approximately satisfied.
	
	\begin{figure}[H]
		\centering
		\includegraphics[width=0.5\linewidth]{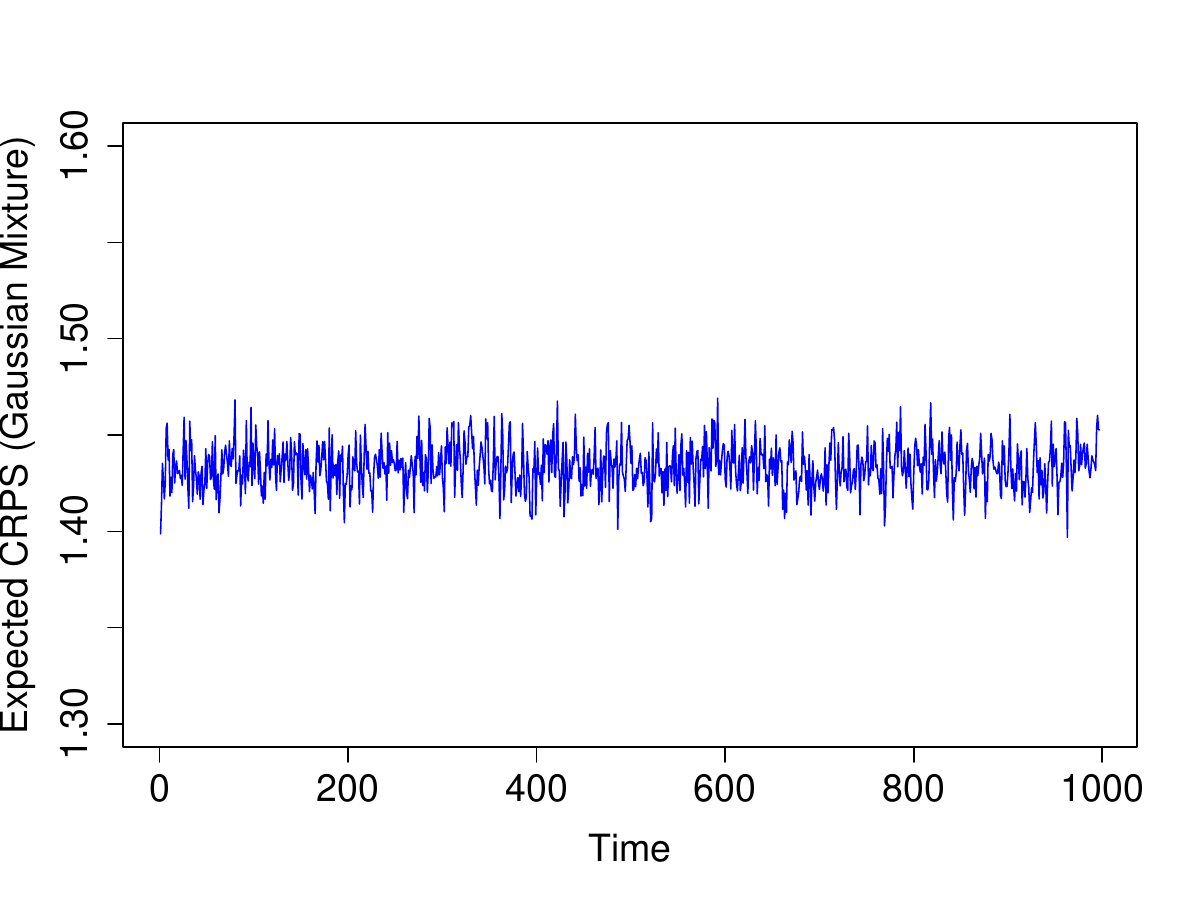}
		\caption{Expected CRPS of the forecaster that corresponds to \((\varepsilon,\delta)=(-0.4, -0.6)\).}
		\label{fig:expected_CRPS}
	\end{figure}
	
	The MCS procedure is applied sequentially in the following way: We start with the original set \(\mathcal{M}_0=\widehat{\M}_0\) containing all \(49\) models.
	During the first \(50\) time points we do not perform any statistical tests.
	Instead, we only observe the outcomes and the report the loss of each forecast.
	This warm-up period is necessary due to the Bootstrap estimation of the variance that the MCS algorithm makes use of.
	Afterwards, for each time point \(t = 51, 52, \ldots, 1000\):
	\begin{enumerate}
		\item We observe the \(t\)-th outcome \(Y_t\).
		\item We test the hypothesis that all remaining models have equal expected loss, namely,
		\begin{equation*}
			H_{0,\widehat{\M}_{t-1}}: \mu_{ij} = 0  \text{ for all } i,j\in \widehat{\M}_{t-1},
		\end{equation*}
		where \(\mu_{ij}=\E[d_{ij,t}]\) denotes the expected loss difference between models \(i\) and \(j\).
		As we mentioned earlier, it is assumed that \(\mu_{ij}\) does not depend on \(t\), so we have omitted the index \(t\).
		We use the test statistic
		\begin{equation*}
			T_{\text{R},\text{M}}=\max_{i,j\in \widehat{\M}_{t-1}} \left|t_{ij}\right|, \text{ where } t_{ij}=\frac{\overline{d}_{ij}}{\sqrt{\widehat{\text{var}}\left(\overline{d}_{ij}\right)}},
		\end{equation*}
		and \(\overline{d}_{ij}=n^{-1}\sum_{s=1}^t d_{ij,s}\) is the average relative loss of models \(i,j\).
		The quantity \(\widehat{\text{var}}\left(\overline{d}_{ij}\right)\) is a Bootstrap estimate of \(\text{var}\left(\overline{d}_{ij}\right)\) based on the first \(t\) observed points.
		\item If \(H_{0,\widehat{\M}_{t-1}}\) is not rejected, then we set \(\widehat{\M}_t=\widehat{\M}_{t-1}\) and \(t\leftarrow t+1\), and we go back to step 1.
		Else, we eliminate the worst-performing model, namely, the one indexed by
		\begin{equation*}
			e_{\text{R},\text{M}}=\argmax_{i\in \widehat{\M}_{t-1}}\left\{\sup_{j\in \widehat{\M}_{t-1}} \frac{\overline{d}_{ij}}{\sqrt{\widehat{\text{var}}\left(\overline{d}_{ij}\right)}}\right\}.
		\end{equation*}
		We then set \(\widehat{\M}_{t-1}\leftarrow \widehat{\M}_{t-1}\backslash \{e_{\text{R,M}}\}\) and go back to step 2.
	\end{enumerate}
	
	For the variance estimation in step 2, we used \(B=2000\) Bootstrap samples.
	For comparison, \cite{MCS} used \(B=1000\) samples in their simulations.
	In the documentation of the \texttt{R}-package \texttt{MCS}, \cite{MCS_package} set \(B=5000\) as the default.
	We tried this option, too, but the computational cost was too high, so we opted for a lower, but still reliable number of Boostrap iterations.
	
	In the above implementation of the MCS algorithm, after a model is eliminated, it is totally discarded in future time points.
	That is, when we set \(t\leftarrow t+1\) and move back to step 1, we only repeat the process for the models in \(\widehat{\M}_t\), and not for all models in \(\M_0\).
	For anytime-valid methods, this is not possible, as illustrated by Equation \eqref{eq:adjusted_e-processes}, where \(m\) stays constant over time.
	
	We implemented the sequential version of the MCS procedure like this, since it leads to lower computational cost.
	As explained earlier, the MCS algorithm involves the Bootstrap estimation \(\widehat{\text{var}}\left(\overline{d}_{ij}\right)\) after \emph{every single model rejection}.
	This is already expensive computationally, and the cost would be higher if we restarted with the full set of models at every new time point.
	
	However, this variation has no impact on the rate of coverage of the optimal model: since we are testing the strong hypothesis, one rejection of this model is enough to exclude it from future confidence sets.
	Even though the losses associated with this model would otherwise be used in future time steps, the model itself would not reenter the model confidence set.
	
	However, the size of the model confidence set could be affected.
	After the removal of a number of models, the definition of \(\overline{d}_{ij}\) changes, and the Bootstrap estimation of the variance is based only on the surviving models.
	Even in this case, the size of the MCS would be pushed down, because it would be easier for sub-optimal models \(i\in \M_0\) to get excluded in future steps.
	This happens because the excluded models \(j\) that are likely worse than \(i\) would otherwise push the quantiles of the distribution of \(\overline{d}_{ij}\) upwards.
	Now that they are not taken into account, performance differences between \(i\) and better-performing models are going to become more pronounced.
	
	As in Figure \ref{fig:Simulations1&2_results}, we track the average size of the confidence set and the rate of coverage of the optimal forecast across a period of \(T = 1000\) time points.
	The averages are computed over \(N=1000\) Monte Carlo iterations and the results are shown in Figure \ref{fig:mcs_vs_smcs}.
	
	As expected, the original MCS has very poor sequential control of the Type I error. In contrast, the anytime-valid SMCSs contains the superior model across the entire time domain.
	Our SMCS are uniformly larger than the (boosted) MCS run sequentially over time. However, this comes at the cost of very poor coverage rates.

	\begin{figure}[H]
		\centering
		\includegraphics[scale = 0.45]{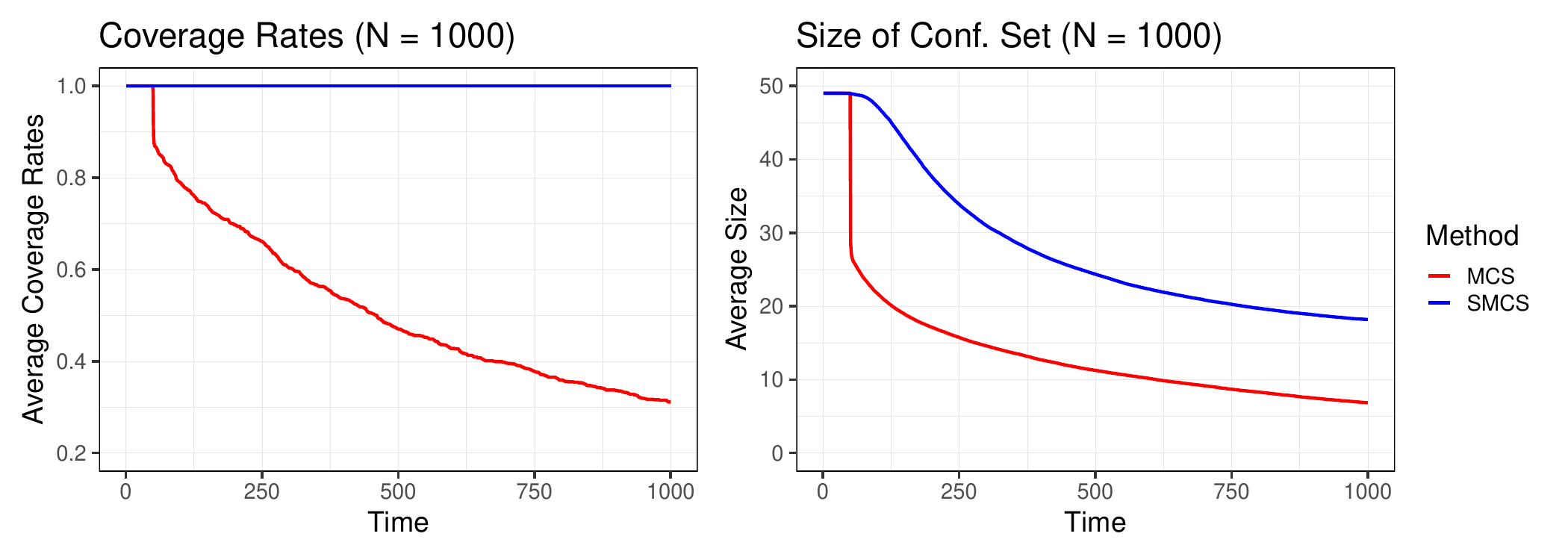}
		\caption{Comparison of the original MCS with SMCS. 
			Left: Coverage rates; Right: Average size of MCS and SMCS.}
		\label{fig:mcs_vs_smcs}
	\end{figure}
	
	\section{Accelerated algorithm for e-value adjustment}\label{app:algo1}
	
	In this section, we present an accelerated version of the e-value adjustment algorithm of \cite{Vovk_Wang_2021}.
	The idea and its rigorous formulation is due to Tim Stephan, a student seminar participant at ETH Z\"urich.
	
	The complexity of the original algorithm \cite[Algorithm~1]{Vovk_Wang_2021} is \(O(m^2)\), where \(m\) is the number of e-values used as an input.
	A quick look at this algorithm (Algorithm \ref{alg: evalue_adjustment}) reveals that the quadratic complexity comes from the nested loop, which runs a total of \(m^2\) iterations.
	
	Observe that the inner loop does not need to run all the way from \(1\) to \(m\).
	Due to the sorting step, the minimum term is attained for some value \(i\in \{0,\ldots,k-1\}\).
	For larger indices, the terms that are added to the numerator are larger than \(e_{\pi(k)}\), so the overall fraction takes higher values.
	This observation can reduce the number of iterations to \(1 + \ldots + m-1\), which is however still quadratic in \(m\).
	
	\begin{algorithm}
		\caption{E-Value Adjustment Algorithm \citep{Vovk_Wang_2021}}
		\label{alg: evalue_adjustment}
		\begin{algorithmic}[1]
			\Require A sequence of e-values \(e_1,\ldots,e_m\).
			\State Find a permutation \(\pi:\{1,\ldots,m\}\to \{1,\ldots,m\}\) such that \(e_{\pi(1)}\leq \ldots e_{\pi(m)}\).
			\State Set \(S_0\vcentcolon = 0\).
			\For{\(i = 1, \ldots, m\)}
			\State \(S_i\vcentcolon = S_{i-1}+e_{\pi(i)}\).
			\EndFor
			\For{\(k=1,\ldots,m\)}
			\State \(e_{\pi(k)}^\star \vcentcolon = e_{\pi(k)}\).
			\For{\(i=1,\ldots,m\)}
			\State \(e\vcentcolon = \frac{e_{\pi(k)}+S_i}{i+1}\)
			\If{\(e<e_{\pi(k)}^\star\)}
			\State \(e_{\pi(k)}^\star \vcentcolon = e\).
			\EndIf
			\EndFor
			\EndFor
			\Ensure Adjusted e-values \(e_1^\star, \ldots,e_m^\star\).
		\end{algorithmic}
	\end{algorithm}
	
	It turns out that we can actually take advantage of the fact that the e-values are first being sorted to bring the number of iterations down to a linear level.
	In that case, the dominating part will be the ordering of the e-values, whose complexity is generally \(O(m\log m)\).
	The accelerated e-value adjustment algorithm relies on the following lemma:
	
	\begin{lem}
		\label{lem: bell_shape}
		Let \(a_1\leq \ldots\leq a_m\) be positive real numbers.
		Set \(S_0 = 0\) and \(S_i=\sum_{k=1}^i a_k\) for \(i=1,\ldots,m\).
		For all \(k\in \{1,\ldots,m\}\), define
		\begin{equation*}
			Q_{k,i}=\frac{a_k + S_i}{i+1},\quad i=0,\ldots,k-1.
		\end{equation*}
		Then:
		\begin{enumerate}
			\item For each \(k\in \{1,\ldots,m\}\), there exists an index \(i_k\in \{1,\ldots,k-1\}\) such that \(Q_{k,i-1}\geq Q_{k,i}\) for all \(i\leq i_k\) and \(Q_{k,i}< Q_{k,i+1}\) for \(i\geq i_k\).
			By convention, we choose \(i_k\) to be the largest index with this property.
			\item For all \(k\in \{1,\ldots,m-1\}\), it holds that \(i_k\leq i_{k+1}\).
		\end{enumerate}
	\end{lem}
	
	\begin{proof}
		The proof is based on the two following straightforward claims.
		Let \(x_1,\ldots,x_m,y\) be arbitrary real numbers and let \(\overline{x}_{1:m}\) denote the arithmetic mean of \(x_1,\ldots,x_m\).
		Then,
		\begin{itemize}
			\item If \(y\geq \overline{x}_{1:m}\), then
			\begin{equation*}
				y\geq \frac{x_1+\ldots+x_m+y}{m+1}\geq \overline{x}_{1:m},
			\end{equation*}
			with equality if and only if \(y= \overline{x}_{1:m}\).
			\item If \(y\leq \overline{x}_{1:m}\), then
			\begin{equation*}
				y\leq \frac{x_1+\ldots+x_m+y}{m+1}\leq \overline{x}_{1:m},
			\end{equation*}
			with equality iff \(y= \overline{x}_{1:m}\).
		\end{itemize}
		We now prove the two parts of the lemma separately.
		\begin{enumerate}
			\item Starting with \(Q_{k,1}\), and since \(a_i\leq a_{i+1}\) for all \(i\in \{1,\ldots,m-1\}\), the sequence \(Q_{k,1},Q_{k,2},\ldots\) is decreasing until we hit an index \(i\in \{1,\ldots,k-1\}\) such that \(a_{i+1}> Q_{k,i}\).
			We call this index \(i_k\).
			If no such index exists, we set \(i_k=k-1\).
			We now need to show that the sequence \(Q_{k,i},Q_{k,i+1},\ldots\) is strictly increasing for \(i\geq i_k\).
			If \(i_k= k-1\), there is nothing to show because \(Q_{k,i}\) is not defined for \(i\geq k\).
			If \(i_k<k-1\), then \(a_{i_k+1}>Q_{k,i_k}\), yields that \(a_{i_k+1}>Q_{k,i_k+1}\).
			The inequality \(Q_{k,i_k} < a_{i_k+1}\) shows that \(Q_{k,i_k}<Q_{k,i_k+1}<a_{i_k+1}\leq a_{i_k+2}\).
			In turn, this yields \(Q_{k,i_k+1}<Q_{k,i_k+2}<a_{i_k+2}\leq a_{i_k+3}\), so the claim follows by induction.
			\item Fix \(k\in \{1,\ldots,m-1\}\).
			If \(i_k=1\), there is nothing to show.
			By definition, it holds that \(Q_{k,i_k-1}\geq Q_{k,i_k}\).
			This implies that \(a_{i_k}\leq Q_{k,i_k-1}\) (this follows by the two facts in the beginning of the proof).
			On the other hand, it holds that \(a_j\leq a_{j+1}\) for all indices \(j\in [k-1]\), so \(a_{i_k}\leq Q_{k,i_k-1}\leq Q_{k+1,i_k-1}\).
			This yields \(Q_{k+1,i_k-1}\geq Q_{k+1,i_k}\), so the turning point of the sequence \(Q_{k+1,1},Q_{k+1,2},\ldots\) is at least as large as \(i_k\).
		\end{enumerate}
	\end{proof}
	
	This lemma reveals that the inner loop needs to run for even fewer iterations.
	Once we reach the value \(i_k\), we can update the value of \(k\) and start the index \(i\) from \(i_k\) instead of starting from \(1\).
	We will update the value of \(k\) again when we reach the point where the sequence \(Q_{k,i}\) becomes strictly increasing.
	This yields a total of \((i_{k}+1)-(i_{k-1}-1)\) iterations for each value of \(k\in \{2,\ldots,m\}\) and a total of \(i_1+1\) iterations for \(k=1\).
	Therefore, the total number of iterations of the nested loop is
	\begin{equation*}
		(i_1+1) + (i_2 + 1 - i_1 +1) + \ldots (i_m + 1-i_{m-1}+1)=i_m+2m-1<3m,
	\end{equation*}
	which is linear in \(m\).
	This idea is summarized in Algorithm \ref{alg: accelerated_adjustment}, which is an accelerated version of Algorithm \ref{alg: evalue_adjustment}.
	
	\begin{algorithm}[H]
		\caption{Accelerated E-Value Calibration}
		\label{alg: accelerated_adjustment}
		\begin{algorithmic}[1]
			\Require A sequence of e-values \(e_1,\ldots,e_m\).
			\State Find a permutation \(\pi:\{1,\ldots,m\}\to \{1,\ldots,m\}\) such that \(e_{\pi(1)}\leq \ldots e_{\pi(m)}\).
			\State Set \(S_0\vcentcolon = 0\).
			\For{\(i = 1, \ldots, m\)}
			\State \(S_i\vcentcolon = S_{i-1}+e_{\pi(i)}\).
			\EndFor
			\State \(i\leftarrow 1\)
			\State \(e_{\pi(1)}^\star \leftarrow e_{\pi(1)}\)
			\For{\(2\leq k\leq m\)}
			\State \(e_{\pi(k)}^\star \leftarrow \frac{e_{\pi(k)}+S_i}{i+1}\).
			\While{\(i<k-1\)}
			\State \(i\leftarrow i+1\).
			\State \(e_{\text{temp}}\leftarrow \frac{e_{\pi(k)}+S_i}{i+1}\).
			\If{\(e_{\text{temp}}\leq e_{\pi(k)}^\star\)}
			\State \(e_{\pi(k)}^\star\leftarrow e_{\text{temp}}\).
			\Else
			\State \(i\leftarrow i-1\).
			\State \textbf{break}
			\EndIf
			\EndWhile
			\EndFor
			\Ensure Adjusted values \(e_1^\star, \ldots,e_m^\star\).
		\end{algorithmic}
	\end{algorithm}
	
\end{document}